\documentclass{tlp}

\usepackage{amsmath, amssymb}
\usepackage{multirow}
\usepackage{setspace}
\usepackage{multicol}

\usepackage{todonotes}
\usepackage{subfigure}
\usepackage{booktabs}

\usepackage{color}
\usepackage{graphicx}
\usepackage[normalem]{ulem}
\newcommand{\change}[1]{
	{#1}}

\newcommand{\rem}[1]{}

\usepackage[ruled,vlined]{algorithm2e}
\SetKwInput{KwInput}{Input}                
\SetKwInput{KwOutput}{Output}
\SetKw{KwGoTo}{go to}

\newcommand{\vir}{``}
\newcommand{\nop}[1]{}

\newcommand{\DC}{\mathsf{C}} 
\newcommand{\DV}{\mathsf{V}} 
\newcommand{\DN}{\mathsf{N}} 
\newcommand{\vars}{\mathit{vars}} 
\newcommand{\uvars}{\mathit{vars_\forall}} 
\newcommand{\exvars}{\mathit{vars_\exists}} 
\newcommand{\front}{\mathsf{vars_\curvearrowright}} 
\newcommand{\head}{\mathit{head}} 
\newcommand{\body}{\mathit{body}} 
\newcommand{\arity}{\mathit{arity}} 
\newcommand{\dom}{\mathit{dom}} 
\newcommand{\lev}{\mathit{lev}}
\newcommand{\const}{\mathit{consts}} 
\newcommand{\preds}{\mathit{preds}} 
\newcommand{\hdpred}{\mathit{h}\textrm{-}\preds} 
\newcommand{\bdpred}{\mathit{b}\textrm{-}\preds} 
\newcommand{\pos}{\mathit{pos}} 
\newcommand{\sch}{\mathit{sch}} 
\newcommand{\mods}{\mathit{mods}} 
\newcommand{\chase}{\mathit{chase}} 
\newcommand{\aff}{\mathit{aff}} 
\newcommand{\nonaff}{\mathit{nonaff}} 
\newcommand{\dang}{\mathit{dang}} 
\newcommand{\harml}{\mathit{harmless}} 
\newcommand{\harmf}{\mathit{harmful}} 
\newcommand{\sigc}{\Sigma_\mathcal{C}} 
\newcommand{\sighg}{\Sigma_\mathrm{HG}} 
\newcommand{\dcompl}{\ensuremath{D^+}\xspace}
\newcommand{\C}{\ensuremath{\mathcal{C}}\xspace} 
\newcommand{\problematic}{\textit{problematic}\xspace} 
\newcommand{\safe}{\textit{safe}\xspace} 
\newcommand{\satoms}{\ensuremath{\mathit{s}\textrm{-}\mathit{atoms}}\xspace}
\newcommand{\patoms}{\ensuremath{\mathit{p}\textrm{-}\mathit{atoms}}\xspace}
\newcommand{\bridge}{\vars_\star} 
\newcommand{\mainrule}{\mathit{main}} 
\newcommand{\hgrule}{\mathit{hg}} 
\newcommand{\Aux}{\mathtt{Aux}} 
\newcommand{\grAt}{\ensuremath{\mathit{gra}}\xspace} 
\newcommand{\grAtDP}{\ensuremath{\mathit{gra}(D,\Pi)}\xspace} 
\newcommand{\HB}{\ensuremath{\mathit{HB}}\xspace} 
\newcommand{\HF}{\ensuremath{\mathit{hf}}\xspace} 

\newcommand{\certEval}{\ensuremath{\textsc{cert-eval}}}
\newcommand{\certEvalC}{\ensuremath{\textsc{cert-eval}[\C]}\xspace}
\newcommand{\certEvalDyaC}{\ensuremath{\textsc{cert-eval}[\dyac]}\xspace}

\newcommand{\dpcertEvalC}{\ensuremath{\textsc{dp-cert-eval}[\C]}\xspace}

\newcommand{\algCompleteC}{\ensuremath{\textsf{Complete}_{[\C]}}\xspace}
\newcommand{\algDpCertEvalC}{\ensuremath{\textsf{DpCertEval}_{[\C]}}\xspace}
\newcommand{\algCertEvalDyaC}{\ensuremath{\textsf{CertEval}_{[\dyac]}}\xspace}

\newcommand{\cert}{\mathit{cert}} 
\newcommand{\dpcert}{\ensuremath{\mathit{dp}\textrm{-}\cert}\xspace} 
\newcommand{\dpchase}{\ensuremath{\mathit{dp}\textrm{-}\chase}\xspace} 

\newcommand{\datalog}{\ensuremath{\textsf{Datalog}}\xspace}
\newcommand{\datex}{\ensuremath{\textsf{Datalog}^{_\exists}}\xspace}
\newcommand{\linear}{\ensuremath{\textsf{Linear}}\xspace}
\newcommand{\indep}{\ensuremath{\textsf{Inclusion-Dependencies}}\xspace}
\newcommand{\joinless}{\ensuremath{\textsf{Joinless}}\xspace}
\newcommand{\sticky}{\ensuremath{\textsf{Sticky}}\xspace}
\newcommand{\stickyj}{\ensuremath{\textsf{Sticky-(Join)}}\xspace}
\newcommand{\prot}{\ensuremath{\textsf{Protected}}\xspace}
\newcommand{\shy}{\ensuremath{\textsf{Shy}}\xspace}
\newcommand{\ward}{\ensuremath{\textsf{Ward}}\xspace}
\newcommand{\joinacy}{\ensuremath{\textsf{Jointly-Acyclic}}\xspace}
\newcommand{\weaklyacy}{\ensuremath{\textsf{Weakly-Acyclic}}\xspace}
\newcommand{\guarded}{\ensuremath{\textsf{Guarded}}\xspace}
\newcommand{\frguarded}{\ensuremath{\textsf{(Fr)-Guarded}}\xspace}
\newcommand{\wguarded}{\ensuremath{\textsf{Weakly-Guarded}}\xspace}
\newcommand{\wfrguarded}{\ensuremath{\textsf{Weakly-(Fr)-Guarded}}\xspace}
\newcommand{\dyac}{\ensuremath{\textsf{Dyadic-}\mathcal{C}}\xspace}
\newcommand{\dfinds}{\ensuremath{\textsf{Af-Inds}}\xspace}
\newcommand{\fes}{\ensuremath{\textsf{FES}}\xspace} 
\newcommand{\fus}{\ensuremath{\textsf{FUS}}\xspace} 
\newcommand{\bts}{\ensuremath{\textsf{BTS}}\xspace} 
\newcommand{\ps}{\ensuremath{\textsf{SPS}}\xspace} 

\newcommand{\ACzero}{\ensuremath{{\footnotesize \textbf{AC}_0}}\xspace}
\newcommand{\NP}{{\footnotesize \textbf{NP}}\xspace}
\newcommand{\PTIME}{{\footnotesize \textbf{PTIME}}\xspace}
\newcommand{\EXPTIME}{{\footnotesize \textbf{EXPTIME}}\xspace}
\newcommand{\TwoEXPTIME}{{\footnotesize \textbf{2EXPTIME}}\xspace}
\newcommand{\iEXPTIME}{{\footnotesize \textbf{i-\EXPTIME}}\xspace}
\newcommand{\iEXPSPACE}{{\footnotesize \textbf{i-EXPSPACE}}\xspace}
\newcommand{\PSPACE}{{\footnotesize \textbf{PSPACE}}\xspace}
\newcommand{\PTIMEc}{\ensuremath{\PTIME{\footnotesize\textrm{-complete}}}\xspace}
\newcommand{\PTIMEh}{\ensuremath{\PTIME{\footnotesize\textrm{-hard}}}\xspace}
\newcommand{\NPTIMEc}{\ensuremath{\NP{\footnotesize\textrm{-complete}}}\xspace}
\newcommand{\NPTIMEh}{\ensuremath{\NP{\footnotesize\textrm{-hard}}}\xspace}
\newcommand{\EXPTIMEc}{\ensuremath{\EXPTIME{\footnotesize\textrm{-complete}}}\xspace}
\newcommand{\TwoEXPTIMEc}{\ensuremath{\TwoEXPTIME{\footnotesize\textrm{-complete}}}\xspace}
\newcommand{\PSPACEc}{\ensuremath{\PSPACE{\footnotesize\textrm{-complete}}}\xspace}

\newtheorem{theorem}{Theorem}
\newtheorem{lemma}{Lemma}
\newtheorem{proposition}{Proposition}
\newtheorem{corollary}{Corollary}
\newtheorem{definition}{Definition}
\newtheorem{example}{Example}

\begin{document}


\lefttitle{Cambridge Author}

\jnlPage{1}{8}
\jnlDoiYr{2021}
\doival{10.1017/xxxxx}

\title[Dyadic Existential Rules]{Dyadic Existential Rules
}

\begin{authgrp}
	\author{\sn{Gottlob} \gn{Georg}}
	\affiliation{Department of Computer Science, University of Oxford, United Kingdom}
	\affiliation{Faculty of Informatics, TU Wien, Austria}
	\author{\sn{Manna} \gn{Marco}, \sn{Marte} \gn{Cinzia}}
	\affiliation{Department of Mathematics and Computer Science, University of Calabria, Italy}
\end{authgrp}

\history{\sub{xx xx xxxx;} \rev{xx xx xxxx;} \acc{xx xx xxxx}}

\maketitle

\newcommand{\mybox}{{\scriptsize\ensuremath{\hfill\blacksquare}\normalsize}}
\begin{abstract}
Existential rules form an expressive $\datalog$-based language to specify ontological knowledge.
The presence of existential quantification in rule-heads, however, makes the main reasoning tasks undecidable.
To overcome this limitation, in the last two decades, a number of classes of existential rules guaranteeing the decidability of query answering have been proposed.
Unfortunately, only some of these classes fully encompass $\datalog$ and, often, this comes at the price of higher computational complexity.
Moreover, expressive classes are typically unable  to exploit tools developed for classes exhibiting lower expressiveness.  
To mitigate these shortcomings,
this paper introduces a novel general syntactic condition that allows us to define, systematically and in a uniform way, from any decidable class $\mathcal{C}$ of existential rules,
a new class called \dyac~enjoying the following properties:
$(i)$ it is decidable;
$(ii)$ it generalises $\datalog$;
$(iii)$ it generalises $\mathcal{C}$; 
$(iv)$ it can effectively exploit any reasoner for query answering over $\mathcal{C}$; and
$(v)$ its computational complexity does not exceed the highest between the one of $\mathcal{C}$ and the one of $\datalog$. Under consideration in Theory and Practice of
Logic Programming (TPLP).
%
\end{abstract}

\begin{keywords} 
	Existential rules,
	Datalog,
	ontology-based query answering,
	tuple-generating dependencies,
	computational complexity.
\end{keywords}


\section{Introduction}
In ontology-based query answering, a conjunctive query is typically evaluated over a logical theory consisting of a relational database paired with an ontology. Description Logics~(\citeauthor{baader2003description} \citeyear{baader2003description}) and Existential Rules ---a.k.a. tuple generating dependencies, or $\datalog^\pm$ rules---(\citeauthor{baget2011rules} \citeyear{baget2011rules}) are the main languages used to specify ontologies. In particular, the latter are essentially classical datalog rules (\citeauthor{abiteboul1995foundations} \citeyear{abiteboul1995foundations}) extended with existential quantified variables in rule-heads.
The presence of existential quantification in the head of rules, however, makes query answering undecidable in the general case.
To overcome this limitation, in the last two decades, a number of classes of existential rules---based on both semantic and syntactic conditions---that guarantee the decidability of query answering have been proposed.
Concerning the semantic conditions, we recall 
{\em finite expansions sets},
{\em finite treewidth sets},
{\em finite unification sets}, and
{\em strongly parsimonious sets}~(\citeauthor{baget2009extending} \citeyear{baget2009extending}; \citeauthor{baget2011rules} \citeyear{baget2011rules}; \citeauthor{Leone2019} \citeyear{Leone2019}).
Each of these classes encompasses a number of concrete classes based on syntactic conditions (\citeauthor{baget2011rules} \citeyear{baget2011rules};
\citeauthor{cali2013taming} \citeyear{cali2013taming};
\citeauthor{fagin2005data} \citeyear{fagin2005data}; 
\citeauthor{krotzsch2011extending} \citeyear{krotzsch2011extending};
\citeauthor{ceri1989you} \citeyear{ceri1989you}; 
\citeauthor{Leone2019} \citeyear{Leone2019};
\citeauthor{gottlob2015beyond} \citeyear{gottlob2015beyond}; 
\citeauthor{protected} \citeyear{protected};
\citeauthor{DBLP:journals/ai/CaliGP12} \citeyear{DBLP:journals/ai/CaliGP12};
\citeauthor{cali2012general} \citeyear{cali2012general};
\citeauthor{gogacz2013converging} \citeyear{gogacz2013converging};
\citeauthor{johnson1984testing}) \citeyear{johnson1984testing}.
Table~\ref{table:classicalComplexity} summarises these classes and their computational complexity with respect to query answering, by distinguishing between \textit{combined complexity} (the input consists of a database, an ontology, a conjunctive query, and a tuple of constants) and 
\textit{data complexity} (only a database is given as input, whereas the remaining parameters are considered fixed).

\begin{table}[t!]
	\caption{\centering Computational complexity of \change{query answering}.}
 \label{table:classicalComplexity}
	{\tablefont\begin{tabular}{@{\extracolsep{\fill}}lcc}
		\topline
		Class \C & Data Complexity & Combined Complexity
		\midline
		$\wfrguarded^{[1],[2]}$
		& $\EXPTIMEc$ & $\TwoEXPTIMEc$ \\
		$\frguarded^{[2]}$
		& $\PTIMEc$ & $\TwoEXPTIMEc$ \\
		$\weaklyacy^{[3]}$
		& $\PTIMEc$ & $\TwoEXPTIMEc$ \\
		$\joinacy^{[4]}$
		& $\PTIMEc$ & $\TwoEXPTIMEc$ \\
		$\datalog^{[5]}$
		& $\PTIMEc$ & $\EXPTIMEc$ \\
		$\shy^{[6]}$
		& $\PTIMEc$ & $\EXPTIMEc$ \\
		$\ward^{[7]}$
		& $\PTIMEc$ & $\EXPTIMEc$ \\
		$\prot^{[8]}$
		& $\PTIMEc$ & $\EXPTIMEc$ \\
		%
		$\stickyj^{[9]}$
		& $\ACzero$ & $\EXPTIMEc$ \\
		$\linear^{[10]}$
		& $\ACzero$ & $\PSPACEc$ \\
		$\joinless^{[11]}$
		& $\ACzero$ & $\PSPACEc$ \\
		$\indep^{[12]}$
		& $\ACzero$ & $\PSPACEc$ \\
	\hline
\end{tabular}}
\small 
$[1]$ \cite{baget2011rules};
$[2]$ \cite{cali2013taming};
$[3]$ \cite{fagin2005data};
$[4]$ \cite{krotzsch2011extending}; $[5]$ \cite{ceri1989you};
$[6]$ \cite{Leone2019};
$[7]$ \cite{gottlob2015beyond}; $[8]$ \cite{protected}; 
$[9]$ \cite{DBLP:journals/ai/CaliGP12};
$[10]$ \cite{cali2012general};
$[11]$ \cite{gogacz2013converging}; 
$[12]$ \cite{johnson1984testing}.
\begin{tabular}{c}
\hline
\hline
\end{tabular}
\end{table}
\normalsize

Unfortunately, on the one side, despite the fact that existential rules generalise datalog rules, only some of these syntactic classes fully encompass \datalog and, in some cases, this even comes at the price of higher computational complexity of query answering.
Moreover, on the other side,   expressive classes typically need \textit{ad hoc} reasoners without being able to exploit mature tools developed for classes exhibiting 
lower expressiveness.

With the aim of mitigating the two aforementioned shortcomings, this paper introduces a novel general syntactic condition that allows to define, systematically and in a uniform way, from any decidable class $\mathcal{C}$ of existential rules, a new class called \dyac~that enjoys the following properties:
$(i)$ it is decidable;
$(ii)$ it generalises \datalog;%
\footnote{\change{Strictly speaking, to guarantee that \dyac generalises \datalog, one has to focus on any $\C \supseteq \dfinds$, where $\dfinds$  is the very simple class of existential rules defined in Section~\ref{sec:dfinds} such that $(a)$ rules are inclusion dependencies with no existential variable and $(b)$ predicates in rule-heads do not appear in rule-bodies.
Indeed, all known classes of existential rules based on semantic conditions as well as all concrete classes reported in Table~\ref{table:classicalComplexity} do encompass $\dfinds$.}}
$(iii)$ it generalises $\mathcal{C}$; and
$(iv)$ it can effectively exploit any reasoner for query answering over $\mathcal{C}$.
In particular, let $\mathbb{C}_d$ (resp., $\mathbb{C}_c$) be the data (resp., combined) complexity of query answering over \C, query answering over \dyac is $\PTIME^{\mathbb{C}_d}$ (resp., $\EXPTIME^{\mathbb{C}_c}$ provided that there is at least an exponential jump from $\mathbb{C}_d$ to $\mathbb{C}_c$).
Since all the classes reported in Table \ref{table:classicalComplexity} comply with the exponential jump assumption, we get the following:
$(a)$ whenever ${\mathbb{C}_d} \supseteq \PTIME$ (entries 1--8 of Table~\ref{table:classicalComplexity}), then query answering over \dyac is complete for $\mathbb{C}_d$ (resp., $\mathbb{C}_c$);
$(b)$ in all the remaining cases (entries 9-12 of Table~\ref{table:classicalComplexity}), query answering over \dyac is complete for $\PTIME$ (resp., $\EXPTIME$), namely it has the same complexity of query answering over $\datalog$.

Concerning the key principle at the heart of this new general syntactic condition,  basically,  an ontology $\Sigma$ belongs to $\dyac$
if one can easily construct a pair $(\sighg, \sigc)$ of ontologies, called \textit{dyadic},  such that: 
$(i)$ $\sighg \cup \sigc$ is equivalent to $\Sigma$ with respect to query answering; 
$(ii)$ $\sigc \in \mathcal{C}$; and 
$(iii)$ $\sighg$ is a set rules called \textit{head-ground} \change{with respect to $\sighg \cup \sigc$}~(\citeauthor{gottlob2015beyond} \citeyear{gottlob2015beyond}).
Intuitively, \change{$\sighg$ satisfies the following properties: $(1)$ it belongs to $\mathsf{Datalog}$;
$(2)$ for each database $D$, the chase procedure~(\citeauthor{DBLP:conf/pods/DeutschNR08} \citeyear{DBLP:conf/pods/DeutschNR08}) over $D \cup \sighg \cup \sigc$ never generates atoms containing null-values via rules of $\sighg$;
$(3)$ head-predicates of $\sighg$ and body-predicates of $\sighg$ are disjoint; and
$(4)$ head-predicates of $\sighg$ and head-predicates of $\sigc$ are disjoint.}
%
%
Finally, since $\dyac$ is well-defined even if $\C$ is a class of existential rules based on some semantic conditions and, if so, since query answering is still decidable over $\dyac$, then
---in analogy with the existing semantic classes---the union of all the \dyac~classes are called {\em dyadic decomposable sets}.

\change{The article is a revised version of an earlier workshop paper (\citeauthor{DBLP:conf/datalog/GottlobMM22} \citeyear{DBLP:conf/datalog/GottlobMM22}).
Specifically, the content that was previously presented in a single preliminary section has been  expanded and reorganised into two longer separate sections, namely Sections~\ref{sec:prelim} and~\ref{sec:TGDs}. 
These sections now contain the necessary background information, ensuring that the paper is self-contained.
Furthermore, the previous  notion of \vir dyadic decomposition'' has evolved into the novel notion of a \vir Dyadic Pair  of TGDs'', which is discussed in Section~\ref{sec:DyadicPairs}.
This new notion captures the essential properties of dyadic decompositions and also generalises the notion of ontology, providing new perspectives and insights.
Additionally, in Section~\ref{sec:DyadicSets}, the notion of \vir Dyadic Decomposable Sets'' is now supported by a canonical concrete algorithm that produces a Dyadic Pair of TGDs from each Dyadic Decomposable Set.
The revisions also lead to new results regarding decidability and complexity.
First, if $\mathcal{C}$ is an abstract (resp., concrete) and decidable class, then $\dyac$ is now also an abstract (resp., concrete) and decidable class.
%
%
Second, the relationship between \datalog and any $\dyac$ is made explicit, emphasising the low expressive power required for $\mathcal{C}$ to ensure that $\dyac$ fully encompasses \datalog.
Finally, the computational complexity analysis is completed in Section~\ref{sec:complex}, where both data and combined complexity for any $\dyac$ class are systematically studied.}



\section{Preliminaries}\label{sec:prelim}
In this section, we introduce the syntax and the semantics of \change{the class of rules that generalises \datalog with existential quantifiers in rule-heads.}
\change{Regarding computational complexity, we assume the reader is familiar with the basic complexity classes used in the subsequent sections: $\ACzero$ $\subseteq$ \PTIME $\subseteq$ \NP $\subseteq$ \PSPACE $\subseteq$  \EXPTIME $\subseteq$ \TwoEXPTIME. Moreover, for a complexity class $\mathbb{C}$, we denote by $\PTIME^{\mathbb{C}}$ (resp., $\EXPTIME^{\mathbb{C}}$) the class of decision problems that can be solved  by an oracle Turing machine operating in polynomial (resp., exponential) \mbox{time with the aid of an oracle that decides a problem in $\mathbb{C}$.}}

\subsection{Basics on Relational Structures} 

Fix three pairwise disjoint lexicographically enumerable infinite sets $\DC$ of {\em constants}, $\DN$ of {\em nulls} ($\varphi$, $\varphi_0$, $\varphi_1$, ...), and $\DV$ of {\em variables} ($x$, $y$, $z$, and variations thereof). Their union is denoted by $\mathsf{T}$ and its elements are called {\em terms}.
For any integer $k \geq 0$, we may write $[k]$ for the set $ \{1,..., k \}$; in particular, as usual, if $k = 0$, then $[k] = \emptyset$.

An {\em atom}  $\underline{a}$ is an expression of the form $P(\textbf{t})$, where $\preds(\underline{a})=P$ is a {\em (relational) predicate}, $\textbf{t}=t_1,..., t_k$ is a tuple of {\em terms} $\arity(\underline{a}) = \arity(P)=k \geq 0$ is the {\em arity} of both $\underline{a}$ and $P$, 
and $\underline{a}[i]$ denotes the $i$-th term $\textbf{t}[i] = t_i$ of $\underline{a}$, for each $i \in [k]$.
In particular, if $k = 0$, then $\textbf{t}$ is the empty tuple and $\underline{a} = P()$.
By $\const(\underline{a})$ and $\vars(\underline{a})$ we denote, respectively, the set of constants and variables occurring in $\underline{a}$.
A {\em fact} is an atom that contains only constants.

A {\em (relational) schema} $\mathbf{S}$ is a finite set of predicates, each with its own arity.
The set of {\em positions} of $\mathbf{S}$, denoted by $\pos(\mathbf{S})$, is defined
as \change{the set} $\{P[i] \ | \ P \in \mathbf{S}~\wedge~ 1 \leq i \leq  \arity(P)\}$, where 
each $P[i]$ denotes the $i$-th {\em position} of $P$.
A {\em (relational) structure} over $\mathbf{S}$ is any (possibly infinite) set of atoms using only predicates from $\mathbf{S}$. 
The {\em domain} of a structure $S$, denoted by $\dom(S)$, is the set of all the terms forming the atoms of $S$.
An {\em instance} over $\mathbf{S}$ is any structure $I$ over $\mathbf{S}$ such that  $\dom(I) \subseteq \DC \cup \DN$. 
A {\em database} over $\mathbf{S}$ is any finite instance over $\mathbf{S}$ containing only facts. 
The {\em active domain} of an instance $I$, denoted by $\dom(I)$, is the set of all the terms occurring in $I$, whereas the {\em Herbrand Base} of $I$, denoted by $\HB(I)$, is the set of all the atoms that can be formed using the predicate symbols of $\mathbf{S}$ and \change{terms of}
 $\dom(I)$.

Consider two sets of terms $T_1$ and $T_2$ and a map $\mu : T_1 \rightarrow T_2$.
Given a set $T$ of terms, the {\em restriction} of $\mu$ with respect to $T$ is the map $\mu|_{T} = \{t \mapsto \mu(t):t \in T_1 \cap T\}$.
An {\em extension} of $\mu$ is any map $\mu'$ between terms, denoted by $\mu' \supseteq \mu$, such that $\mu'|_{T_1} = \mu$.
A \textit{homomorphism} from a structure $S_1$ to a structure $S_2$ is any map $h:  \mathit{dom}(S_1) \rightarrow \mathit{dom}(S_2)$ such that 
both the following hold: $(i)$ if $t \in \DC \cap \mathit{dom}(S_1)$, then $h(t) = t$; and $(ii)$ $h(S_1) = \{P(h(\mathbf{t})) : P(\mathbf{t}) \in S_1\} \subseteq S_2$.

\subsection{Conjunctive Queries} 

A {\em conjunctive query} (CQ) $q$ over a schema $\mathbf{S}$ is a (first-order) formula of the form
\begin{equation}\label{eq:query}
	\langle \mathbf{x} \rangle \leftarrow \exists \ \mathbf{y} \ \Phi(\mathbf{x,y}),    
\end{equation}
\noindent where $\mathbf{x}$ and $\mathbf{y}$ are tuples (often seen as sets) of variables such that $\mathbf{x} \cap \mathbf{y} = \emptyset$, and $\Phi(\mathbf{x,y})$ is a conjunction (often seen as a set) of atoms using only predicates from $\mathbf{S}$.
In particular,
\begin{itemize}
\vspace{-1.5mm}
    \item[-] $\mathit{dom}(\Phi) \subseteq \mathbf{x} \cup \mathbf{y} \cup \DC$,
    \item[-] \change{whenever a variable $z$ belongs to $\mathbf{x} \cup \mathbf{y}$, then $z$ occurs also in $\Phi$,}
    \item[-] $\mathbf{x}$ are the {\em output} variables of $q$, and 
    \item[-] $\mathbf{y}$ are the {\em existential} variables of $q$.
\vspace{-1.5mm}
\end{itemize}
To highlight the output variables, we may write $q(\mathbf{x})$ instead of $q$.
The {\em evaluation} of $q$ over an instance $I$ is the set $q(I)$ of every tuple $\mathbf{t}$ of constants admitting a homomorphism $h_{\mathbf{t}}$ from $\Phi(\mathbf{x,y})$ to $I$ such that $h_{\mathbf{t}}(\mathbf{x}) = \mathbf{t}$.

A \textit{Boolean conjunctive query} (BCQ) is a CQ with no output variable, 
namely an expression of the form 
\mbox{$\langle \rangle \leftarrow \exists \ \mathbf{y} \ \Phi(\mathbf{y})$}. An instance $I$ satisfies a BCQ $q$, denoted $I \models q$, if $q(I)$ is nonempty, namely $q(I)$ contains only the empty tuple $\langle \rangle $.

\subsection{Tuple-Generating Dependencies}

A \textit{tuple-generating dependency} (TGD) $\sigma$---also known as {\em (existential)  rule}---over a schema $\mathbf{S}$ is a (first-order) formula of the form
\begin{equation}\label{eq:rule}
	\Phi (\mathbf{x}, \mathbf{y}) \rightarrow \ \exists \ \mathbf{z} \ \Psi(\mathbf{x}, \mathbf{z}),
\end{equation}
\noindent where $\mathbf{x}$, $\mathbf{y}$, and $\mathbf{z}$ are pairwise disjoint tuples of variables, and both $\Phi(\mathbf{x,y})$ and $\Psi(\mathbf{x,z})$ are conjunctions (often seen as a sets) of atoms using only predicates from $\mathbf{S}$.
In particular, 
\begin{itemize}
\vspace{-1.5mm}
\item[-] $\Phi$ (resp., $\Psi$) contains all and only the variables in $\mathbf{x} \cup \mathbf{y}$ (resp., $\mathbf{x} \cup \mathbf{z}$), 
\item[-] constants (but not nulls) may also occur in $\sigma$, 
\item[-] $\mathbf{x} \cup \mathbf{y}$ are the {\em universal} variables of $\sigma$ denoted by $\uvars(\sigma)$, 
\item[-] $\mathbf{z}$ are the {\em existential} variables of $\sigma$ denoted by $\exvars(\sigma)$, and 
\item[-] $\mathbf{x}$ are the {\em frontier} variables of $\sigma$ denoted by $\front(\sigma)$.
\vspace{-1.5mm}
\end{itemize}
\noindent We refer to $\body(\sigma) = \Phi$ and $\head(\sigma) = \Psi$ as the {\em body} and {\em head} of $\sigma$, respectively.
If $|\head(\sigma)| = 1$, the TGD is called \textit{single-head}, otherwise it called \textit{multi-head}.
If $\exvars(\sigma) = \emptyset$ and $|\head(\sigma)|=1$, then $\sigma$ is called $\mathit{datalog}$ rule.
With $ \hdpred(\sigma) $ (resp., $\bdpred(\sigma)$) we denote the set of predicates in $\head(\sigma)$ (resp., $\body(\sigma)$).
\change{An instance $I$ satisfies $\sigma$, written $I \models \sigma$, if the existence of 
a homomorphism $h$ from $\Phi$ to $I$ implies the existence of a homomorphism $h' \supseteq h_{| \mathbf{x}}$ from $\Psi$ to $I$.}

An \textit{ontology} $\Sigma$ is a set of rules. \change{An instance  $I$ satisfies $\Sigma$, written $I \models \Sigma$, if $I \models \sigma$ for each $ \sigma \in \Sigma $.}
Without loss of generality, we assume that $\vars(\sigma_1) \cap \vars(\sigma_2)~=~\emptyset$, for each pair $\sigma_1, \sigma_2$ of rules in $\Sigma$.
Operators $\exvars$, $\hdpred$, and $\bdpred$ naturally extend on ontologies.

A class $\mathcal{C}$ of  ontologies is any (typically infinite) set of TGDs fulfilling some syntactic or semantic conditions (see, for example, the classes shown in Table~\ref{table:classicalComplexity}, some of which will be formally defined in the subsequent sections).
In particular, $\mathsf{Datalog}$ is the class of ontologies containing only datalog rules.

\change{Finally,} the \textit{schema} of an ontology $\Sigma$, denoted $\sch(\Sigma)$, is the subset of $ \mathbf{S}$ containing all and only the predicates occurring in $\Sigma$, whereas $ \arity(\Sigma) = \max_{P \in \sch(\Sigma)} \arity(P)$.
For simplicity of exposition, we write $\pos(\Sigma)$ instead of $\pos(\sch(\Sigma))$.

\rem{An instance $I$ satisfies a rule $\sigma$ as in Equation~\ref{eq:rule}, written $I \models \sigma$, if the existence of a homomorphism $h$ from $\Phi$ to $I$ implies the existence of a homomorphism $h' \supseteq h_{| \mathbf{x}}$ from $\Psi$ to $I$.}
\rem{An instance  $I$ satisfies a
set $ \Sigma $ of TGDs, written $ I \models \Sigma $, if $ I \models \sigma $ for each $ \sigma \in \Sigma $.}

\subsection{Ontological Query Answering}
Consider a database $D$ and a set $\Sigma$ of TGDs.
A \textit{model} of $D$ and $\Sigma$ is an instance
$I$ such that $I \supseteq D$ and $I \models \Sigma$. Let $\mods(D, \Sigma)$ be the set of all models of $D$ and $\Sigma$. 
The \textit{certain answers}
to a CQ $q$ w.r.t. $ D $ and $\Sigma$ are defined as the set of tuples
$ \cert(q, D, \Sigma) = \bigcap_{M \in \mods(D,\Sigma)} q(M).$
Accordingly, for any fixed schema $\mathbf{S}$, two ontologies $\Sigma_1$ and $\Sigma_2$ over $\mathbf{S}$ are said to be $\mathbf{S}$-\textit{equivalent}  (in symbols $\Sigma_1 \equiv_\mathbf{S} \Sigma_2$) if, for each $D$ and $q$ over $\mathbf{S}$, it holds that $\cert(q,D,\Sigma_1) = \cert(q,D,\Sigma_2).$
The pair $D$ and $\Sigma$ 
satisfies a BCQ $q$, written $D \cup \Sigma \models q$, if $\cert(q, D, \Sigma) = \langle \rangle$, namely $M \models q$ for each $M \in \mods(D,\Sigma)$.
Fix a class $\mathcal{C}$ of ontologies.
The computational problem studied in this work---called \certEvalC---can be schematized as follows:

\vspace{3mm}

\begin{center}
	\fbox{\hspace{1.2cm}
		\begin{minipage}{10cm}
			\begin{itemize}
				\item[\textsc{Problem:}] \certEvalC.
				\item[\textsc{Input:}] A database $D$, a  ontology $\Sigma \in \C$,
				a conjunctive query $q(\mathbf{x})$, and a tuple $\mathbf{c} \in \DC^{|\mathbf{x}|}$.
				\item[\textsc{Question:}] Does $\mathbf{c} \in \cert(q, D, \Sigma)$ hold?
			\end{itemize}
	\end{minipage}}
\end{center}

\vspace{3mm}

In what follows, 
with a slight abuse of terminology, whenever we say that \C is decidable, we mean that \certEvalC is decidable.
Note that $\mathbf{c} \in \cert(q, D, \Sigma)$ if, and only if, 
$D \cup \Sigma \models q(\mathbf{c})$,
where $q(\mathbf{c})$ is the BCQ obtained from $q(\mathbf{x})$ by replacing, for each $i \in \{1,...,|\mathbf{x}|\}$, every occurrence of the variable $\mathbf{x}[i]$ with the constant $\mathbf{c}[i]$.
Actually, the former problem is \ACzero reducible to the latter.

While considering 
the computational complexity of \certEvalC, we recall the following convention: 
$(i)$ {\em combined complexity} means that $D$, $\Sigma$, $q$, and $\mathbf{c}$ are given in input; and 
$(ii)$ {\em data complexity} means that only $D$ and $\mathbf{c}$ are given in input, whereas $\Sigma$ and $q$ are considered fixed.
\change{Accordingly, we point out that  complexity results reported in Table~\ref{table:classicalComplexity} refer to \certEvalC under this convention.}

\subsection{The Chase Procedure} 

The chase procedure (\citeauthor{DBLP:conf/pods/DeutschNR08} \citeyear{DBLP:conf/pods/DeutschNR08}) is a tool exploited for reasoning with TGDs.
Consider a database $D$ and a set $\Sigma$ of TGDs.
Given an instance $I \supseteq D$,
a {\em trigger} for $I$ is any pair $\langle \sigma, h\rangle$, where $\sigma \in \Sigma$ is a rule as in Equation~\ref{eq:rule} and $h$ is a  homomorphism from $\body(\sigma)$ to $I$.
Let $I' = I \cup h'(\head(\sigma))$, where $h' \supseteq h|_{\mathbf{x}}$ maps each $z \in \exvars(\sigma)$ to a \vir fresh'' null $h'(z)$ not occurring in $I$ such that $z_1 \neq z_2$ in $\exvars(\sigma)$ implies $h'(z_1) \neq h'(z_2)$.
Such an operation which constructs $I'$ from $I$ is called 
\textit{chase step} and denoted $\langle \sigma, h \rangle (I) = I'$.

Without loss of generality, we assume that nulls introduced at each trigger functionally depend on the pair  $ \langle \sigma, h \rangle $ that is involved in the trigger.
For example, given a rule $\sigma$ as in Equation \ref{eq:rule} and a homomorphism $h$, it is sufficient to pick $ \varphi_{\langle \mathbf{z},h(\mathbf{x},\mathbf{y}) \rangle} $ as the fresh null replacing $\mathbf{z}$ when the chase produces the trigger $\langle \sigma, h \rangle$.
Accordingly, the processing order of rules and triggers does not change the result of the chase, and hence $\chase(D, \Sigma$) can be considered unique.
The \textit{chase procedure} of $D \cup \Sigma$ is an
exhaustive application of chase steps, starting from $D$, which produce a sequence $I_0 = D \subset I_1 \subset I_2 \subset \dots \subset I_m \subset \dots$ of instances in such a way that:
$(i)$ for each $i\geq 0$, $I_{i+1} = \langle \sigma, h \rangle (I_{i})$ is a chase step obtained via some trigger $\langle \sigma, h \rangle$ for $I_{i}$;
$(ii)$ for each $i \geq 0$, if there exists a trigger $\langle \sigma, h \rangle$ for $I_i$, then there exists some $j>i$ such that $I_{j} = \langle \sigma, h\rangle (I_{j-1})$ is a chase step; and
$(iii)$ any {\em trigger} $\langle \sigma, h\rangle$ is used only once.
We define $\chase(D,\Sigma) = \cup_{i\geq0} I_i$.

The {\em chase bottom} is the finite set of all null-free atoms in $\chase(D,\Sigma)$ and is defined as $\chase^{\bot} (D, \Sigma) = \chase(D, \Sigma) \cap \HB(D)$.

It is well know that $\chase(D,\Sigma)$ is a {\em universal model} of $D \cup \Sigma$, that is, for each $M \in \mods(D, \Sigma)$ there is a homomorphism from $\chase(D,\Sigma)$ to $M$. 
Hence, given a BCQ $q$ it holds that $\chase(D,\Sigma) \models q \Leftrightarrow D \cup \Sigma \models q$.

We recall that $\chase(D, \Sigma)$ can be decomposed into \textit{levels} (\citeauthor{cali2010advanced} \citeyear{cali2010advanced}): 
each atom of $D$ has level
$\gamma = 0$; an atom of $\chase(D,\Sigma)$ has level $\gamma + 1$ if, during its generation, the exploited trigger $\langle \sigma, h \rangle$ maps the body of $\sigma$ via $h$ to atoms whose maximum level is $\gamma$. We refer to the
part of the chase {\em up to} level $\gamma$ as $\chase^\gamma(D, \Sigma)$.
Clearly, $\chase(D,\Sigma) = \cup_{\gamma \geq 0} \chase^\gamma(D, \Sigma)$.
Finally, a trigger {\em involved} at a certain level $j$ if it gives rise to an atom of level $j$.

\section{Considered Decidable Classes of TGDs}\label{sec:TGDs}
In this section we provide an overview of the main existing decidable classes of TGDs.
We recall both syntactic and semantic classes, where the first are based on a specific syntactic condition that can be checked, while the latter are classes that do not come with a syntactic property that can be checked on rules and, hence, are not recognizable.
\change{Finally, we introduce a very simple new class of existential rules called \dfinds. We will exploit the latter to sharpen our results presented in Section \ref{sec:DyadicSets} and Section \ref{sec:complex}.}

\subsection{Preliminary Notions}

We start fixing some basics notions.
We have chosen to provide a uniform notation for the key existing notions of affected and invaded positions, such as attacked, protected, harmless, harmful, and dangerous variables~(\citeauthor{Leone2019} \citeyear{Leone2019};
\citeauthor{cali2013taming} \citeyear{cali2013taming};
\citeauthor{krotzsch2011extending} \citeyear{krotzsch2011extending};
\citeauthor{nuovoAndreas} \citeyear{nuovoAndreas};
\citeauthor{DBLP:conf/datalog/GottlobMM22} \citeyear{DBLP:conf/datalog/GottlobMM22}).
Basically, these notions serve to separate positions in which the chase can introduce only constants from those where nulls might appear.

\begin{definition}[S-affected positions]\label{S-affected}
Consider an ontology $\Sigma$ and a variable  $ z \in \exvars(\Sigma) $.
A position $\pi \in \mathsf{pos}(\Sigma)$ is {\em $z$-affected} (or {\em invaded} by $z$) if one of the following two properties holds:
$(i)$ there exists $ \sigma \in \Sigma $ such that $ z $ appears in the head of $ \sigma $ at position $\pi$; 
$(ii)$ there exist $ \sigma \in \Sigma $ and $ x \in \mathsf{front}(\sigma) $ such that $x$ occurs both in $ \head(\sigma) $ at position $\pi$ and in $ \body(\sigma) $ at $z$-affected positions only.
Moreover, a position $ \pi \in \mathsf{pos}(\Sigma)$ is $ S $-{\em affected}, where $ S \subseteq \exvars(\Sigma) $, if:
$(i)$ for each $ z \in S $, $ \pi $ is $z$-affected; and
$(ii)$ for each $ z \in \exvars(\Sigma) $, if $ \pi $ is $ z $-affected, then $ z \in S $. \mybox
\end{definition}

\noindent Note that for every position $ \pi $ there exists a unique set $S$ such that $\pi$ is $S$-affected. We write $\aff(\pi)$ for this set $S$. Moreover, 
\mbox{$ \aff(\Sigma) = \{\pi \in \pos(\Sigma) \ | \ \aff(\pi) \neq \emptyset  \}, $}
and
$ \nonaff(\Sigma) = \pos(\Sigma) \setminus \aff(\Sigma).$

\medskip
\noindent We point out that the notion above presented is a refined version of the classical notion of affected position (\citeauthor{cali2013taming} \citeyear{cali2013taming}).
In particular, it holds that if a position $\pi$ is $S$-affected, then $\pi$ is also affected; whereas if $\pi$ is affected, then $\pi$ may not be $S$-affected.
Moreover, the $S$-affected notion coincides with the one of {\em attacked positions by a variable} (\citeauthor{Leone2019} \citeyear{Leone2019}; \citeauthor{krotzsch2011extending} \citeyear{krotzsch2011extending}).
We highlight that its key nature and properties are not modified by the notion of $S$-affected position introduced above. Hence, for simplicity of exposition, \change{we give only}\rem{we find more useful to give only} this refined definition.
In the same spirit, we classify variables occurring in a conjunction of atoms.

\begin{definition}[Variables classification]\label{def:variable_classification}
Let $\sigma$ be a TGD and $ x \in \vars(\body(\sigma)) $. Then,
$(i)$ if $ x $ occurs at positions $ \pi_1, \dots, \pi_n $ and $ \bigcap_{i=1}^{n} \mathsf{aff}(\pi_i) = \emptyset $, $ x $ is said \change{to be} $\harml$;
$(ii)$ if $ x $ is not harmless, \change{with}\rem{placed} $S =  \bigcap_{i=1}^{n} \mathsf{aff}(\pi_i) $, it is said \change{to be} \textit{$ S $-harmful};
$(iii)$ if $ x $ is $ S $-harmful and belongs to $ \front(\sigma) $, $ x $ is $ S $-\textit{dangerous}.\mybox
\end{definition}

\noindent Given a variable $ x $ that is $S$-dangerous, we write $ \dang(x) $ for the set $ S $.
Hereinafter, for simplicity of exposition, the prefix $S$- is omitted when it is not necessary.
Consider an ontology $\Sigma$. Given a rule $\sigma \in \Sigma$, we denote by $\dang(\sigma)$ (resp., $\harml(\sigma)$ and $\harmf(\sigma)$) the dangerous (resp., harmless and harmful) variables in $\sigma$. %
These sets of variables naturally extend to the whole $\Sigma$  by taking, for each of them, the union over all the rules of $\Sigma$. 

\subsection{Decidable Classes of Existential Rules}

We now survey \rem{about} the fifteen concrete classes reported in Table~\ref{table:classicalComplexity} as well as the known abstract classes based on semantic conditions.
On the one side, we report \change{some} specific syntactic conditions whenever these are useful for the rest of the presentation; on the other side, \change{for all of them (both concrete and abstract)}, we recall their  containment relationships.
\mbox{For the rest of the section, fix a $\datex$ ontology $\Sigma$.}

\medskip 
The class $\fes$ (\citeauthor{baget2009extending} \citeyear{baget2009extending}) stands for \textsf{finite expansions sets}, which intuitively are sets of TGDs which ensure the termination of the chase procedure.
The class $\bts$ (\citeauthor{baget2009extending} \citeyear{baget2009extending}) stands for \textsf{bounded treewidth sets}, which intuitively are sets of TGDs which guarantee that the (possibly infinite) instance constructed by the chase procedure has bounded treewidth.
The class $\fus$ (\citeauthor{baget2011rules} \citeyear{baget2011rules}) stands for \textsf{finite unification sets}, which intuitively are sets of TGDs which guarantee the termination of (resolution-based) backward chaining procedures.
The class $\ps$ (\citeauthor{Leone2019} \citeyear{Leone2019}) stands for \textsf{strongly parsimonious sets}, which intuitively are sets of TGDs  which guarantee that the parsimonious chase procedure can be reapplied  a number of times that is linear
in the size of the query.

\medskip


\medskip
The class $\weaklyacy$ (\citeauthor{fagin2005data} \citeyear{fagin2005data}) is based on the {\em acyclicity} condition. To define the latter, we recall the \textit{label graph} of $\Sigma$, $G(\Sigma) = \langle N,A\rangle$, defined as follows:
$(i)$ $N = \cup_{P \in \sch(\Sigma)} \pos(P)$;
\change{$(ii)$ $(\pi_1, \pi_2, \forall) \in A$ if there are $\sigma \in \Sigma$ and $x \in \front(\sigma)$ such that $x$ occurs both in $\body(\sigma)$ at position $\pi_1$ and in $\head(\sigma)$ at position $\pi_2$; and
$(iii)$ $(\pi_1, \pi_2, \exists) \in A$ if there are $\sigma \in \Sigma$, $x \in \front(\sigma)$, and $y \in \exvars(\sigma)$ such that both $x$ occurs in $\body(\sigma)$ at position $\pi_1$ and $y$ occurs in $\head(\sigma)$ at position $\pi_2$.}
The \textit{existential graph} of $\Sigma$ is
$G_\exists(\Sigma) = \langle N,A\rangle$, where:
$(i)$ $N = \cup_{\sigma \in \Sigma}\exvars(\sigma)$; and
$(ii)$ $(X,Y) \in A$ if the rule $\sigma$ where $y$ occurs contains a universal variable $x$-affected and occurring in $\head(\sigma)$.
Therefore, $\Sigma$ belongs to $\weaklyacy$ (resp., $\joinacy$) if 
$G(\Sigma)$ (resp., $G_\exists(\Sigma)$) has no cycle going through an $\exists$-arc (resp., is acyclic).


\medskip
We now recall the notion of marked variable, in order to define the class $\sticky$ (\citeauthor{DBLP:journals/ai/CaliGP12} \citeyear{DBLP:journals/ai/CaliGP12}).
A variable $x$ of $\Sigma$ is {\em marked} if
$(i)$ there is $\sigma \in \Sigma$ such that $x$ occurs in $\body(\sigma)$ but not in $\head(\sigma$); or
$(ii)$ there is $\sigma \in \Sigma$
such that $x$ occurs in $\head(\sigma)$ at position \change{$\pi$} together with
some $\sigma' \in \Sigma$ having a marked variable in its body at position \change{$\pi$}.
Accordingly, the {\em stickiness} condition states that $\Sigma$ is $\sticky$ if, for each $\sigma \in \Sigma$, 
$x$ occurs multiple times in $\body(\sigma)$ implies $x$ is not marked.

\medskip
The class $\linear$ (\citeauthor{cali2012general} \citeyear{cali2012general}) is based on the {\em linearity} condition: an ontology $\Sigma$ belongs to
$\linear$ if each rule contains at most one body atoms.
This class generalize the class $\indep$ (\citeauthor{abiteboul1995foundations} \citeyear{abiteboul1995foundations}; \citeauthor{johnson1984testing} \citeyear{johnson1984testing}) in which rules contain only one body atom and one head atom and the repetition of variables is not allowed neither in the body nor in the head.

\medskip
The class $\guarded$ (\citeauthor{cali2013taming} \citeyear{cali2013taming}) is based on the {\em guardedness} condition: an ontology $\Sigma$ belongs to $\guarded$ if for each rule $\sigma \in \Sigma$ there is $\underline{a}$ in $\body(\sigma)$ such that $\uvars(\sigma) = \vars(\underline{a})$.
In similar fashion, $\Sigma$ belongs to $\wguarded$ if, for each $\sigma \in \Sigma$, there is an atom of $\body(\sigma)$ containing all the affected variables~of~$\sigma$.

\medskip
We recall the {\em shyness} condition underlying the class $\shy$ (\citeauthor{Leone2019} \citeyear{Leone2019}).
An ontology $\Sigma$ is $\shy$ if, for each $\sigma \in \Sigma$ the following conditions both hold: $(i)$ if a variable $x$ occurs in more than one body atom, then $x$ is harmless; $(ii)$ for every pair of distinct dangerous variable $z$ and $w$ in different atoms, $ \dang(z) \cap \dang(w) = \emptyset$.

\medskip
The class $\ward$ (\citeauthor{gottlob2015beyond} \citeyear{gottlob2015beyond}) is based on the {\em wardedness} condition: $\Sigma \in \ward$ if, for each $\sigma \in \Sigma$, there are no dangerous variables in $\body(\sigma)$, or there exists an atom $\underline{a} \in \body(\sigma)$, called a \textit{ward}, such that
$(i)$ all the dangerous variables in $\body(\sigma)$ occur in $\underline{a}$, and
$(ii)$ each variable of $\vars(\underline{a}) \cap \vars(\body(\sigma) \setminus \{\underline{a}\})$ is harmless.

\medskip
\change{Having finished with syntactic and semantic conditions, we close the section with a proposition stating their containment relationships (\citeauthor{baget2011rules} \citeyear{baget2011rules}; 
\citeauthor{krotzsch2011extending} \citeyear{krotzsch2011extending}; \citeauthor{Leone2019} \citeyear{Leone2019}; \citeauthor{protected} \citeyear{protected}).}

\begin{proposition}\label{prop:basicCont}
The following classes are pairwise uncomparable, except for:
\begin{itemize}
\item[] \begin{itemize}
    \item[-] $\indep \subset \joinless$, $\indep \subset \linear$;
    \item[-] $\joinless \subset \sticky \subset \sticky$-$\mathsf{Join} \subset \fus$;
    \item[-] $\linear \subset \guarded$, $\linear \subset \mathsf{Protected}$;
    \item[-] $\guarded \subset \mathsf{Weakly}$-$\guarded$, $\guarded \subset \mathsf{Fr}$-\guarded;
    \item[-] $\mathsf{Weakly}$-$\guarded \subset \mathsf{Weakly}$-$\mathsf{Fr}$-\guarded $\subset \bts$;
    \item[-] $\mathsf{Fr}$-\guarded $\subset \mathsf{Weakly}$-$\mathsf{Fr}$-\guarded;
    \item[-] $\datalog \subset \mathsf{Weakly}$-$\guarded$,
          $\datalog \subset \mathsf{Protected}$,
          $\datalog \subset \weaklyacy$;
    \item[-] $\mathsf{Protected} \subset \ward$, $\mathsf{Protected} \subset \shy \subset \ps$;
    \item[-] $\weaklyacy \subset \joinacy \subset \fes$.
\end{itemize}
\end{itemize}
 \end{proposition}

\change{Throughout the remainder of the paper, let $\mathbb{E}_{syn}$ denote the set of all fifteen decidable syntactic classes reported in Table~\ref{table:classicalComplexity}. 
Analogously, 
let $\mathbb{E}_{sem}$ denote the set of known decidable abstract classes considered in this paper, namely $\fes$, $\fus$, $\bts$, and $\ps$.}


\bigskip
\subsection{Autonomous Full Inclusion Dependencies}\label{sec:dfinds}
The aim of this section is to introduce a very simple \change{new} class of existential rules called $\dfinds$. Additionally, we characterise the main properties of this class.

\begin{definition}[\dfinds]\label{def:dfinds}
An ontology $\Sigma$ belongs to $\dfinds$ (autonomous full inclusion dependencies) if $\Sigma$ belongs to $\indep$ and the following conditions are also satisfied: (1) head predicates do not appear in bodies (autonomous property);
(2) rules have no existential variables (full property).
\end{definition}

\noindent Now, we show that any class \C of TGDs in $\mathbb{E}_{syn} \cup \mathbb{E}_{sem}$ includes the class just defined.
Formally, it holds the following.

\begin{proposition}\label{prop:dfindsSubseteqC}
Consider a class $\C \in \mathbb{E}_{syn} \cup \mathbb{E}_{sem}$ of TGDs. Then, $\dfinds \subseteq \C$.
\end{proposition}

\vspace{3mm}
\begin{proof}
Thanks to Proposition~\ref{prop:basicCont}, the statement becomes equivalent to show that $(i)$ $\dfinds \subseteq \indep$ and $(ii)$ $\dfinds \subseteq \datalog$.
%
%
By Definition~\ref{def:dfinds}, the class $\dfinds$ contains all the rules that have only one body and head atom, without repetition of variables neither in the body nor in the head, and that satisfy the autonomous property (head atom does not appear in bodies) and the full property (rules have only one head atom without existential variables). 
Accordingly, relation $(i)$ and $(ii)$ are trivially full-filled.\end{proof}

\noindent We conclude the section by providing the complexity of the class \dfinds.

\begin{proposition}
\mbox{$\certEval[\dfinds]$} is in \ACzero in data complexity and  \NPTIMEc in combined complexity.
\end{proposition}

\begin{proof}
By Proposition \ref{prop:dfindsSubseteqC}, $\dfinds \subseteq \indep$. 
Hence, the data complexity of the problem $\certEval[\dfinds]$ is inherit from that of $\certEval[\indep]$, that is $\ACzero$.
For the combined complexity, we first observe that the problem $\certEval[\dfinds]$ is $\NPTIMEh$, \rem{since every conjunctive query can be seen as an ontology in $\dfinds$ having as rule the query itself.}\change{building upon the well-known fact that $\certEval[\emptyset]$ is  already $\NPTIMEh$. The latter refers to the problem of evaluating a query against a database in the absence of an ontology.}
Secondly, to prove the completeness of the $\certEval[\dfinds]$ problem, we show that given a query $q(\mathbf{x})$ and an ontology $\Sigma$,  it is possible to construct in $\NP$ a CQ $q_\Sigma(\mathbf{x})$ such that $c \in cert(D,\Sigma,q)$ iff $c \in q_\Sigma(D)$, with $\mathbf{c}$ being a tuple in $\DC^{|\mathbf{x}|}$.
To this aim, for each atom $\underline{a} \in q(\mathbf{x})$, we guess if leave $\underline{a}$ unchanged, or \vir resolv" $\underline{a}$ with the body of some rule $\sigma$ in $\Sigma$ such that $\head(\sigma)$ unify with $\underline{a}$.
Accordingly, $q_\Sigma(D)$ is polynomial with respect to the input and, finally, it is possible to guess in $\NP$ an homomorphism to check if $c \in q_\Sigma(D)$.
\end{proof}


\medskip

\section{Dyadic Pairs of TGDs}\label{sec:DyadicPairs}

In this section we lay the groundwork for the main contribution of the paper, that is the definition of a new decidable class of TGDs called $\dyac$.
%
To this aim we first introduce some preliminary notions in order to define a {\em dyadic pair} and, then we conclude with some computational properties.

\subsection{Formal Definition}
We start introducing the concept of {\em head-ground} set of rules, being roughly \vir non-recursive'' rules in which nulls are neither created nor propagated.

\begin{definition}[Head-ground rules] \label{def:headGround}
Consider an ontology $ \Sigma $. A set $\Sigma' \subseteq  \Sigma $ is  {\em head-ground} w.r.t. $ \Sigma $ if the following are true:
$(1)$ $ \Sigma' \in \mathsf{Datalog}$;
$(2)$ each head atom of \ $ \Sigma'$ contains only harmless variables w.r.t. $ \Sigma$;
\mbox{$(3)$ $ \hdpred(\Sigma') \cap \bdpred(\Sigma') = \emptyset $}; and
$(4)$ $ \hdpred(\Sigma') \cap \hdpred(\Sigma \setminus \Sigma') = \emptyset $.	\mybox
\end{definition}


\noindent The following example is given to better understand the above definition.
\bigskip
\begin{example}
Consider the next set of rules:
$$\begin{array}{rrcl}
    \sigma_1: & R(x_1, y_1), \change{S(y_1,u_1), T(u_1,v_1)} & \rightarrow & \exists \ z_1, w_1 \ Q(z_1, w_1)\\
    \sigma_2: & C(y_2), R(x_2, z_2)  & \rightarrow & S(y_2, z_2) \\
    \sigma_3: & D(y_3,z_3), R(x_3, w_3) & \rightarrow & T(x_3, y_3)\\
    \sigma_4: & Q(x_4, y_4) & \rightarrow & \exists \ z_4 A(x_4,z_4) \\
    \sigma_5: & A(x_5,z_5), D(y_5,z_5) & \rightarrow & Q(x_5,y_5)
\end{array}$$

\noindent A subset of head ground rule w.r.t. $\Sigma$ is given by $ \sighg = \{ \sigma_2, \sigma_3 \}$.
In fact, 
$\harml(\Sigma)$ is the set $\{x_1, y_1, y_2, x_2, z_2, x_3,y_3, z_3, y_5, z_5\}$; hence, according to Definition \ref{def:headGround}, it is easy to check that
$(i)$ $\sigma_2$ and $\sigma_3$ are datalog rules; $(ii)$ the head atoms of $\sigma_2$ and $\sigma_3$ contain only harmless variables; \rem{$(iii)$ both predicates that appear in $\head(\sigma_2)$ and $\head(\sigma_3)$ do not occur in any body of \ $\sighg$, $(iv)$ nor in the head of rules $\sigma_1, \sigma_4$ and $\sigma_5$.}
\change{$(iii)$ both predicates $S$ and $T$ do not occur in the body of any rule in $\sighg$, and $(iv)$ both predicates $S$ and $T$ do not occur in the head of any rule in $\{\sigma_1, \sigma_4, \sigma_5\}$.}
\rem{On the contrary, rules $\sigma_1, \sigma_4$ and $\sigma_5$ could not be in $\sighg$, since they violate Properties 2 and 3 of Definition \ref{def:headGround}. Hence, we observe that the set $\sighg$ is maximal.}
\change{On the contrary, none of the rules in $\{\sigma_1, \sigma_4, \sigma_5\}$ can be part of any head-ground subset of $\Sigma$. Indeed, according to  Definition \ref{def:headGround}, both $\sigma_1$ and $\sigma_2$ violate properties (1) and (2), whereas $\sigma_5$ violates property $(2)$. Hence, we observe that the set $\sighg$ is also maximal.}\mybox
\end{example}

Having in mind the notion of head-ground set of rules, we can now formally define what is a dyadic pair.

\begin{definition}[Dyadic pairs]\label{def:DyadicDec}
	Consider a class \C of TGDs. 
	A pair $\Pi = (\sighg, \sigc) $ of TGDs is {\em dyadic}  with respect to $\mathcal{C}$ if the next hold: 
	$(1)$ $ \sighg $ is head-ground with respect to $ \sighg \cup \sigc$; and 
	$(2)$ $\sigc \in \mathcal{C}$. \mybox
\end{definition}

Whenever the above definition applies, we also say, for short, that $\Pi$ is a $\C$-dyadic pair.
Consider the following example to more easily understand the concept of dyadic pair.

\begin{example}\label{ex:runningExample}
Consider the following pair $\Pi = (\sighg, \sigc) $ of TGDs, where $\sighg$ is:
\begin{center}
	$\begin{array}{rcl}
		P(x_1) & \rightarrow & H_1(x_1)\\
		P(x_2) & \rightarrow & H_2(x_2)\\
		Q(x_3) & \rightarrow & H_3(x_3).\\
	\end{array}$
\end{center}
and $\sigc$ is:
\begin{center}
	$\begin{array}{rcl}
		H_1(x_1) & \rightarrow & \exists \ y_1, z_1 \ R(y_1,z_1) \\
		H_2(x_2) & \rightarrow & \exists \ y_2 \ Q(y_2) \\
		R(y_3,x_3),H_3(x_3) & \rightarrow & S(y_3).\\
	\end{array}$
\end{center}

\noindent In particular, $\Pi$ is a dyadic pair with respect to any $\C \in \{\guarded, \shy, \ward.\}$.
To this aim, let $\Sigma = \sighg \cup \sigc$.
It easy computable that $\aff(\Sigma) = \{R[1],R[2],Q[1],S[1]\}$, where
$\aff(R[1])=\{y_1\}$,
$\aff(R[2])=\{z_1\}$,
$\aff(Q[1])=\{y_2\}$, and
$\aff(S[1])=\{y_1\}$.
Accordingly,
$\harml(\Sigma) = \{x_1,x_2,x_3\}$,
$\harmf(\Sigma) = \{y_3\}$ and $ \dang(\Sigma) = \{y_3\}$.
To prove that $\Pi$ is a dyadic pair, we have first to show that $\sighg$ is an head ground set of rules with respect to $\Sigma$.
Clearly, $\sighg \in \datalog$ and each head atom contains only harmless variables; moreover, the head predicates do not appear neither in body atoms of $\sighg$ nor in head atoms of $\sigc$. 
Hence, $\sighg$ is head-ground with respect to $\Sigma$.
It remains to show that $\sigc \in \C$. We focus on the last rule of $\sigc$, since the first two rules are linear rules, and hence are trivially guarded, shy and ward rules.
The last rule belongs to $\guarded$ since the atom $R(y_3,x_3)$ contains all the universal variables of the rule (guardedness condition); 
it belong to $\shy$ since the variable $x_3$ that occurs in two body atoms is harmless (shyness condition); 
finally, it belongs to $\ward$ since atom $R(y_3,x_3)$ is the ward that contains the dangerous variables $(y_3)$ and shares with the rest of the body only harmless variables $(x_3)$ (wardedness condition). \mybox
\end{example}

The next step is to extend the query answering problem---classically defined over an ontology---over a dyadic pair.
Therefore, we extend both notions of chase and certain answers for a dyadic pair.
Accordingly, given a dyadic pair $\Pi = (\sighg, \sigc)$, we define
\begin{equation}\label{eq:chaseDP}
\dpchase(D,\Pi) = \chase(D, \sighg \cup \sigc)
\end{equation}
and
\begin{equation}\label{eq:certDP}
\dpcert(q,D,\Pi) = \cert(q,D,\sighg \cup \sigc).
\end{equation}
%

\medskip
\noindent Now we can fix the problem studied in the rest of the paper.

\begin{center}
\fbox{\hspace{1.2cm}
\begin{minipage}{10cm}
\begin{itemize}
	\item[\textsc{Problem:}] \dpcertEvalC.
	\item[\textsc{Input:}] A database $D$, a $\C$-dyadic pair $\Pi=(\sighg, \sigc)$ of TGDs, a CQ $q(\mathbf{x})$, and a tuple $\mathbf{c} \in 
	\mathsf{C}^{|\mathbf{x}|}$.
	\item[\textsc{Question:}] Does  $\mathbf{c} \in \dpcert(q,D,\Pi)$ holds?
\end{itemize}
\end{minipage}}
\end{center}


\subsection{Computational Properties}

For the rest of the section, fix a decidable class $\C$ of TGDs.
Given a database $D$ and a $\C$-dyadic pair  $\Pi$ of TGDs, we define the following set of ground atoms:
\begin{equation}\label{eq:D'}
	\grAt(D,\Pi) = \{\underline{a} \in \chase(D, \Pi)~|~ \Pi = (\sighg,\sigc) \ \wedge \  \preds(\underline{a}) \in \hdpred(\sighg) \}.
\end{equation}

\noindent Our idea is to reduce query answering over a dyadic pair $\Pi$ to query answering over $\mathcal{C}$, the latter being decidable by assumption.

\begin{theorem}\label{thm:equivCertEval}
Consider a database $D$, a $\C$-dyadic pair $\Pi = (\sighg,\sigc)$ of TGDs, and a conjunctive query $q(\mathbf{x})$. 
Let $\dcompl = D \cup \grAt(D,\Pi)$.
It hods that $\dpcert(q,D, \Pi) = \cert(q,\dcompl,\sigc)$.
\end{theorem}

\begin{proof}
\rem{Consider $\Pi = (\sighg,\sigc)$.
By definition \mbox{$\dpcert(q,D,\Pi) = \cert(q,D,\sighg \cup \sigc)$};
accordingly, the thesis boils down to showing that 
\mbox{$\cert(q,D,\sighg \cup \sigc) = \cert(q,\dcompl,\sigc)$}.
Fix any arbitrary $|{\bf x}|$-ary tuple ${\bf c}$ of constants and let $q' = q({\bf c})$.}
\change{Consider $\Pi = (\sighg,\sigc)$. By equation (\ref{eq:certDP}), \mbox{$\dpcert(q,D,\Pi) = \cert(q,D,\sighg \cup \sigc)$}.
Moreover, by fixing any arbitrary $|{\bf x}|$-ary tuple ${\bf c}$ of constants, it holds that ${\bf c} \in \cert(q,D,\sighg \cup \sigc) \Leftrightarrow D \cup \sighg \cup \sigc \models q({\bf c})$ and 
${\bf c} \in \cert(q,\dcompl,\sigc) \Leftrightarrow \dcompl \cup \sigc \models q({\bf c})$.
Let $q' = q({\bf c})$.
Accordingly, the thesis boils down to showing that 
$$D \cup \sighg \cup \sigc \models q' \Leftrightarrow \dcompl \cup \sigc \models q'.$$}

\rem{To show that $\cert(q,D, \sighg \cup \sigc) \subseteq \cert(q,\dcompl,\sigc)$, it suffices to prove that \mbox{$D \cup \sighg \cup \sigc \models q'$ implies $\dcompl \cup \sigc \models q'$}.}
$[\Rightarrow]$ Assume that $D \cup \sighg \cup \sigc \models q'$ holds. 
%
%
Hence,  $\chase(D , \sighg \cup \sigc) \models q'$.
Given that $\grAt(D,\Pi) \subseteq \chase(D, \sighg \cup \sigc)$, it holds that $\chase(D \cup \grAt(D,\Pi) , \sighg \cup \sigc) \models q'$.
Moreover, since $\grAt(D,\Pi)$ contains all the auxiliary ground consequences of $\sighg$, the latter becomes equivalent to \mbox{$\chase(D \cup \grAt(D,\Pi) , \sigc) \models q'$}.
Hence, \mbox{$D \cup \grAt(D,\Pi) \cup \sigc \models q'$},
that is $\dcompl \cup \sigc \models q'$.

\rem{To show that $\cert(q,D, \sighg \cup \sigc) \supseteq \cert(q,\dcompl,\sigc)$ it suffices to prove that if \mbox{$\dcompl \cup \sigc \models q'$}, then $D \cup \sighg \cup \sigc \models q'$.}
$[\Leftarrow]$ Assume that $\dcompl \cup \sigc \models q'$, hence $\chase(\dcompl , \sigc) \models q'$.
Since $\sigc \subseteq \sighg \cup \sigc$, it holds that $\chase(\dcompl, \sighg \cup \sigc) \models q'$. 
By hypothesis, $\grAt(D,\Pi) \subseteq \chase(D , \sighg \cup \sigc)$;
hence $\chase(D , \sighg \cup \sigc) \models q'$, that is $D \cup \sighg \cup \sigc \models q'$.
\end{proof}

\begin{algorithm}[h!]
	\DontPrintSemicolon
	\LinesNumbered
	\KwInput{A database $D$, a dyadic pair $\Pi$, a CQ $q(\mathbf{x})$, and a tuple $\mathbf{c} \in 
		\mathsf{C}^{|\mathbf{x}|}$}
	$\dcompl := \algCompleteC(D,\Pi)$; \\
	\Return{$\mathbf{c} \in \cert(q,\dcompl,\sigc)$};
\caption{$\algDpCertEvalC(q,D,\Pi,\mathbf{c})$}\label{alg:DpCertEvalC}
\end{algorithm}

According to Theorem~\ref{thm:equivCertEval}, 
to solve $\dpcertEvalC$ one can first \vir complete'' $D$ and then performing classical query evaluation.
To this aim we design Algorithm~\ref{alg:DpCertEvalC} and Algorithm~\ref{alg:DBcomplete}.
The correctness of the latter will be proved in Proposition~\ref{prop:AlgCorretness}. Consequently, the correctness of the former is guaranteed by Theorem~\ref{thm:equivCertEval}.

In particular, given a database $D$ and a dyadic pair $\Pi = (\sighg, \sigc)$, Algorithm~\ref{alg:DBcomplete} iteratively constructs the set $\dcompl = D \cup \grAt(D,\Pi)$, with $\grAt(D,\Pi)$ being the set defined by Equation \ref{eq:D'}.
Roughly speaking, the first two instructions are required, respectively, to add $D$ to $\dcompl$ and to initialise a temporary set $\tilde{D}$ used to store ground consequences derived from $\sighg$.
The rest of the algorithm is an iterative procedure that computes the certain answers (instruction 5) to the queries constructed from the rules of $\sighg$ (instruction 4) and completes the initial database D (instruction 7) until no more auxiliary ground atoms can be produced (instruction 6).
We point out that, in general, $\tilde{D} \subseteq \grAt(D,\Pi)$ holds; %
in particular, $\tilde{D} = \grAt(D,\Pi)$ holds in the last execution of
instruction 7 or, equivalently, when the condition $D \cup \tilde{D} \supset \dcompl$ examined at instruction 6 is false, since all the auxiliary ground atoms have been added to $\dcompl$.

Before we prove that Algorithm~\ref{alg:DBcomplete} always terminates and correctly constructs $\dcompl$, we show the following preliminary lemma.

\begin{algorithm}[h!]
	\DontPrintSemicolon
	\LinesNumbered
	\KwInput{A database $ D $ and a $\C$-dyadic pair $\Pi = (\sighg, \sigc)$}
	\KwOutput{The set $\dcompl$ of ground atoms} 
	$\dcompl := D$\; 
	$\tilde{D} := \emptyset$\; \label{step2}
	\For{$\mathrm{each~rule~of~the~form \ } \Phi(\mathbf{x,y}) \rightarrow \mathtt{H}_i(\mathbf{x}) \mathrm{ \ in \ } \sighg$} 
	{
		$ q:= \langle \mathbf{x} \rangle \leftarrow \Phi(\mathbf{x,y})$\;
		$\tilde{D} = \tilde{D} \cup \{ \mathtt{H}_i(\mathbf{t}) \ | \ \mathbf{t} \in \cert(q,\dcompl,\sigc) \} $\;
	}
	\If{$(D \cup \tilde{D} \supset \dcompl)$}
	{\label{Step6}
		$\dcompl := D \cup \tilde{D}$\; \label{Step7}
		\KwGoTo instruction \ref{step2}\;
	}
	\Return{$\dcompl$}
	\caption{$\algCompleteC(D,\Pi)$}\label{alg:DBcomplete}
\end{algorithm}


\begin{lemma}\label{lem:chaseContainment}
Consider a database $D$ and a set $\Sigma$ of TGDs.
Let $\Sigma' \subseteq \Sigma$ and $X \subseteq \chase^\bot(D,\Sigma) \setminus D$.
Then, $\chase(D \cup X, \Sigma') \subseteq \chase(D, \Sigma)$.
\end{lemma}

\begin{proof}
Let $X = \{\underline{a}_1,...,\underline{a}_n \}$. 
For each $j \in [n]$, let $\lev(\underline{a}_i)$ be the level of $\underline{a}_i$ inside $\chase(D, \Sigma)$.
Let $p = \max_{j \in [n]}\{\lev(\underline{a}_i)\}$.
The proof proceeds by induction on the level $i$ of $\chase(D \cup X, \Sigma')$.

\medskip

\textit{Base case: $i = 1$}. We want to prove that $\chase^1(D \cup X, \Sigma') \subseteq \chase(D, \Sigma)$.
Let $\underline{a}$ be an atom of $\chase^1(D \cup X, \Sigma')$ generated exactly at level $i = 1$. 
By definition, $\underline{a}$ is obtained due to some trigger $\langle \sigma, h \rangle$ such that $h$ maps $\body(\sigma)$ to $D \cup X$. 
If $\underline{a} \in \chase^p(D, \Sigma)$, then the claim holds trivially. Otherwise, we can show that $\underline{a} \in \chase^{p+1}(D, \Sigma)$.
Indeed, since $h$ maps $\body(\sigma)$ to $D \cup X$, then 
$h$ is also a trigger involved at level $p+1$ since it  maps $\body(\sigma)$ to $\chase^{p}(D, \Sigma)$. 
In particular, $h$ maps at least one atom of $\body(\sigma)$ to some atom $\underline{a}_k \in X$ such that $\lev(\underline{a}_k) = p$.
Since, by definition, nulls introduced during the chase functionally depend on the involved triggers, then $\underline{a}$ necessarily belongs to $\chase^{p+1}(D, \Sigma)$.

\medskip

\textit{Induction step: $i = \ell$.} Given that for every level $i \leq \ell -1$, $\chase^{i}(D \cup X, \Sigma') \subseteq \chase(D,\Sigma)$ (\textit{induction hypothesis}), we prove that $ \chase^{\ell}(D \cup X, \Sigma') \subseteq \chase(D,\Sigma) $ holds, too.
Let $\beta$ be an atom of $ \chase^{\ell}(D \cup X, \Sigma') $ generated exactly at level $i=\ell$.
By definition, $\beta$ is obtained via some trigger $ \langle \sigma', h' \rangle $ such that $h'$ maps $\body(\sigma')$ to atoms with level at most $\ell-1$.
Accordingly, by induction hypothesis, $h'$ maps $\body(\sigma')$ also to $\chase(D,\Sigma)$.
Hence, since the processing order of rules and triggers does not change the result of the chase and nulls functionally depend on the involved triggers, it follow that also $\beta \in \chase(D,\Sigma)$. \end{proof}

With the next proposition, we prove that Algorithm \ref{alg:DBcomplete} always terminates and correctly constructs $\dcompl$.

\begin{proposition}\label{prop:AlgCorretness}
	Consider a database $D$ and a $\C$-dyadic pair $\Pi$ of TGDs.
	It holds that
	Algorithm \ref{alg:DBcomplete} both terminates and computes $\dcompl = D \cup \grAt(D,\Pi)$.
\end{proposition}

\begin{proof}	
Let $\Pi = (\sighg, \sigc)$.
We proceed by proving first the termination of Algorithm \ref{alg:DBcomplete} and then its correctness. 

\medskip

\noindent {\bf Termination.} To prove the termination of Algorithm \ref{alg:DBcomplete}, it suffices to show that each instruction alone always terminates and that the overall procedure never falls into an infinite loop.
First, observe that $|\grAtDP| \leq |\hdpred(\sighg)|\cdot d^\mu$, where
$d=|\const(D)|$ and $ \mu = \max_{P \in \hdpred(\sighg)}{\arity(P)}$.
Instructions 1, 2, 4, 8 and 9 trivially terminate.
Instructions 6 and 7 both terminate, since
$\tilde{D} \subseteq \grAtDP$ always holds (see correctness below). 
Each time instruction 3 is reached, the {\bf for}-loop simply scans the set $\sighg$, which is finite by definition. 
Concerning instruction 5, it suffices to observe that its termination relies on the termination of $\certEvalC$---which is true by hypothesis---and on the fact that, for each query $q$, to construct the set $\{ \mathtt{H}_i(\mathbf{t}) \ | \ \mathbf{t} \in \cert(q,\dcompl,\sigc) \}$, the problem $\certEvalC$ must be solved at most $d^\mu$ times, \change{where $d^\mu$ is the maximum}\rem{being the maximum} number of tuples {\bf t} for which the check $\mathbf{t} \in \cert(q,\dcompl,\sigc)$ has to be performed.
Since each instruction alone terminates, it remains to analyze the overall procedure. 
It contains two loops. The first, namely the {\bf for}-loop at instruction 3, is not problematic; indeed, we shown that it locally terminates. The second one, namely the {\bf go to}-loop, depends on the evaluation of the {\bf if}-instruction, which can be executed at most $|\grAtDP|$ times. Thus, also the {\bf go to}-loop does the same.

\medskip

\noindent {\bf Correctness.}
We now claim that Algorithm \ref{alg:DBcomplete} correctly completes the database. 
Let $\dcompl$ be the output of Algorithm \ref{alg:DBcomplete}.
Our claim is that $\dcompl = D \cup \grAtDP$.

\medskip

Inclusion 1 ($\dcompl \supseteq D \cup \grAtDP$).
Assume, by contradiction, that $D \cup \grAtDP$ contains some atom that does not belong to $\dcompl$.
This means that there exists some $j>0$ such that both $\bar{D} = ((D \cup \grAtDP) \cap \chase^{j-1}(D, \sighg \cup \sigc)) \subseteq \dcompl$ and
$((D \cup \grAtDP) \cap \chase^j(D, \sighg \cup \sigc)) \setminus \dcompl \neq \emptyset$ hold.
Thus, there exists some $\underline{a} \in \chase^j(D, \sighg \cup \sigc)$ whose level is exactly $j$ and that does not belong to $\dcompl$.
Let $\langle \sigma, h\rangle$ be the trigger used by the chase to generate $\underline{a}$, where $\sigma$ is of the form $\Phi(\mathbf{x,y}) \rightarrow \mathtt{H}(\mathbf{x})$. 
Clearly, $h$ maps $\Phi(\mathbf{x},\mathbf{y})$ to $\chase^{j-1}(D, \sighg \cup \sigc)$, and we also have that $\underline{a} = \mathtt{H}(h(\mathbf{x}))$.
Consider now the query $q =\langle \mathbf{x} \rangle \leftarrow \Phi(\mathbf{x,y})$ constructed from $\sigma$ by Algorithm~\ref{alg:DBcomplete} at instruction 4.
Thus, $\chase^{j-1}(D, \sighg \cup \sigc) \models q(h({\bf x}))$ holds. 
Since $\bar{D} \subseteq \dcompl$, we have that $\chase^{j-1}(D, \sighg\cup\sigc) \subseteq  
\chase(\bar{D}, \sigc)
\subseteq
\chase(\dcompl, \sigc)$. 
Hence, $\chase(\dcompl, \sigc) \models q(h({\bf x}))$, namely $h({\bf x}) \in \cert(q,\dcompl,\sigc)$ and, thus, $\underline{a} \in \dcompl$, which is a contradiction.

\medskip 

Inclusion 2 ($\dcompl \subseteq D \cup \grAtDP$).
Let \dcompl be the set produced by Algorithm \ref{alg:DBcomplete}.
Let $\ell$ be the number of time instruction 7 of Algorithm~\ref{alg:DBcomplete} is executed.
At each execution $i \in [\ell]$ of instruction 7, the algorithm computes the set $\tilde{D}_i$ containing  only auxiliary ground atoms, and produces the set $\dcompl_i = D \cup \tilde{D}_i$.
By construction, $\dcompl = D \cup \tilde{D}_\ell$.
Let $\tilde{D}_0 = \emptyset$, $\dcompl_0 = D \cup \tilde{D}_0$, and $I_i = \tilde{D}_i \setminus \tilde{D}_{i-1}$, for each $i \in [\ell]$.
We show that $D \cup \tilde{D}_\ell \subseteq D \cup \grAtDP$, that is $\tilde{D}_\ell \subseteq \grAtDP$.
We proceed by induction on the number $\ell$ of iterations.

\textit{Base case:} Let $i = 1$. We claim that $ \tilde{D}_1 \subseteq \grAtDP $.
By construction, the set
$\tilde{D}_1  =   \{ \mathtt{H}_i(\mathbf{t}) ~|~ i \in [|\sighg|] ~\wedge~ \mathbf{t} \in \cert(q, \dcompl_0, \sigc) \}$.
Since $\dcompl_0 = D$ and the component $\sigc$ does not produce any atom in the first iteration of the algorithm, the latter is equal to 
$\{ \mathtt{H}(\mathbf{t}) ~|~ i \in [|\sighg|] ~\wedge~ \mathbf{t} \in \cert(q, D, \emptyset) \}
= \{ \underline{a} \in \chase^1(D, \sighg) ~|~ \preds(\underline{a}) \in \hdpred(\sighg) \}
\subseteq \{ \underline{a} \in \chase(D, \sighg \cup \sigc) ~|~ \preds(\underline{a}) \in \hdpred(\sighg) \} = \grAtDP$.

\textit{Induction step:} 
Given that, for $i = \ell-1, \ \tilde{D}_{\ell -1} \subseteq \grAtDP $ (\textit{induction hypothesis}), we prove that $ \tilde{D}_\ell \subseteq \grAtDP $ holds, too.
By construction
$\tilde{D}_\ell = \{ \mathtt{H}_i(\mathbf{t}) ~|~ i \in [|\sighg|] ~\wedge~ \mathbf{t} \in \cert(q, \dcompl_{\ell-1}, \sigc) \}
= \{ \mathtt{H}_i(\mathbf{t}) ~|~ i \in [|\sighg|] ~\wedge~ \mathbf{t} \in \cert(q, D \cup \tilde{D}_{\ell-1}, \sigc) \}
= \{ \underline{a} \in \chase(D \cup \tilde{D}_{\ell-1}, \sighg \cup \sigc) ~|~ \preds(\underline{a}) \in \hdpred(\sighg) \}$.
Since the set $ D \cup \tilde{D}_{\ell-1}$ already contains all the ground consequences of $\sighg$, the latter is equivalent to
$\{ \underline{a} \in \chase(D \cup \tilde{D}_{\ell-1}, \sigc) ~|~ \preds(\underline{a}) \in \hdpred(\sighg) \}$.
Applying Lemma \ref{lem:chaseContainment} together with the induction hypothesis, the last one is a subset of
$ \{ \underline{a} \in \chase(D, \sighg \cup \sigc) ~|~ \preds(\underline{a}) \in \hdpred(\sighg) \} = \grAtDP$.
\end{proof}

\noindent We conclude the section by proving the decidability of the problem \dpcertEvalC, under the assumption that \certEvalC is decidable.

\begin{theorem}
If $\certEvalC$ is decidable, then $\dpcertEvalC$ is decidable\rem{too}.
\end{theorem}

\begin{proof}
To prove the decidability of $\dpcertEvalC$ we
design Algorithm~\ref{alg:DpCertEvalC}.
Let $D$ be a database, $\Pi=(\sighg,\sigc)$ a dyadic pair, $q(\mathbf{x})$ a CQ, and $\mathbf{c}$ a tuple in $\DC^{\mathbf{x}}$.
Clearly, step 1 always terminates since it recalls Algorithm~\ref{alg:DBcomplete} that, as shown in Proposition~\ref{prop:AlgCorretness}, always terminates and correctly constructs the set $\dcompl$.
By Theorem~\ref{thm:equivCertEval}, checking if $\mathbf{c} \in \dpcert(q,D,\Pi)$ boils down to checking if $\mathbf{c} \in \cert(q,\dcompl,\sigc)$, that is decidable by hypothesis. Hence, step 2 never falls in a loop and Algorithm~\ref{alg:DpCertEvalC} correctly computes $\dpcertEvalC$.
\end{proof}

\section{Dyadic Decomposable Sets}\label{sec:DyadicSets}

In this section we introduce a novel general condition that allows to define, from any decidable class $\mathcal{C}$ of ontologies, a new decidable class called \dyac enjoying desirable properties.
The union of all the $\dyac$ classes, with $\mathcal{C}$ being any decidable class of TGDs, forms what we call {\em dyadic decomposable sets}, which encompass and generalize any other existing decidable class, including those based on semantic conditions.

We start the section by providing a classification of atoms of the rule-body, according to where dangerous variables appear; then we define the class $\dyac$ proving that the query answer problem over this class is decidable.

\begin{definition}[Atoms classification]\label{def:atomsClassification}
	
Consider a set $\Sigma$ of TGDs and a rule $\sigma \in \Sigma$.
An atom $\underline{a}$ of $\body(\sigma)$ is $\sigma$-\textit{\problematic} if
$(i)$ $\underline{a}$ contains a dangerous variable w.r.t. $\Sigma$, or
$(ii)$ $\underline{a}$ is connected to some $\sigma$-\problematic atom via some harmful variable.
The set of all the \problematic atoms of $\sigma$ is denoted by $\patoms(\sigma)$, 
whereas \mbox{$\satoms(\sigma) = \body(\sigma) \setminus \patoms(\sigma)$} denotes the set of all the \safe atoms of $\sigma$. \mybox
\end{definition}

\noindent We highlight that $\patoms$ and $\satoms$ can share only harmless variables. The next example is to clarify the above definition.

\begin{example}\label{ex:problematicAtoms}
Consider the database $D=\{L(a),R(a,b)\}$, and the following set $\Sigma$ of TGDs:

$$\begin{array}{rrcl}
\sigma_1: & L(x_1) & \rightarrow & \exists ~ y_1 ~ P(y_1,x_1) \\
\sigma_2: & P(x_2,y_2) & \rightarrow & \exists ~ z_2 ~ Q(y_2,z_2,x_2) \\
\sigma_3: & P(z_3,x_3), Q(x_3,y_3,z_3) & \rightarrow & S(z_3)\\
\sigma_4: & P(x_4,y_4), Q(y_4,z_4,w_4), S(w_4), R(u_4,v_4) & \rightarrow & T(x_4,z_4,u_4).
\end{array}$$

\

\noindent We focus on rule $\sigma_4$. Taking into account that $\aff(\Sigma)=\{P[1],S[1],Q[2],Q[3]\}$, it follows that $\dang(\Sigma)=\{x_2,x_4,z_3,z_4\}$, $\harmf(\Sigma)=\{x_2,x_4,y_3,z_3,z_4,w_4\}$, and 
$\harml(\Sigma)=\{x_1,x_3,y_2,y_4,u_4,v_4\}$.
Hence, $\patoms(\sigma_4) = \{P(x_4,y_4), Q(y_4,z_4,w_4), S(w_4)\} $, whereas $\satoms(\sigma_4) = \{R(u_4,v_4)\}$. \mybox
\end{example}

The second step consists in selecting variables shared by \patoms and \satoms together with the harmless frontier variables appearing in \satoms.
The latter can be expressed via the set
$\HF(\sigma) = \{ x \in \vars(\satoms\change{(\sigma)}) ~|~ x \in \harml(\sigma) ~\wedge~ x \in \front(\sigma) \}$.
Then, we define $\bridge(\sigma) = \{\vars(\patoms\change{(\sigma)}) ~\cap~ \vars(\satoms\change{(\sigma)})\} ~ \cup ~ \HF(\sigma)$.
%
Trivially, it holds that $\bridge(\sigma) \subseteq \vars(\satoms\change{(\sigma)})$.
Finally, let 
%

\[
\hgrule(\sigma) : \satoms(\sigma) \, \rightarrow \, \Aux_{\sigma}(\bridge(\sigma)),
\]

\[
\mainrule(\sigma) : \Aux_{\sigma}(\bridge(\sigma)), \ \patoms(\sigma) \, \rightarrow \, \exists \, \exvars(\sigma) \ \head(\sigma),
\]

\[
\hgrule(\Sigma) = \bigcup_{\sigma \in \Sigma} \hgrule(\sigma)
\ \ \ \ \ \ \ \ \textrm{and} \ \ \ \ \ \ \ \ 
\mainrule(\Sigma) = \bigcup_{\sigma \in \Sigma} \mainrule(\sigma).
\]

\medskip

By considering again Example \ref{ex:problematicAtoms}, we have that $\bridge(\sigma_4)=\{u_4,y_4\}$ and also that
$\mainrule(\sigma_4): \Aux_{\sigma_4}(u_4,y_4),P(x_4,y_4), Q(y_4,z_4,w_4), S(w_4,y_4) \rightarrow T(x_4,z_4,u_4)$.

In the special case in which a variable $x$ of $\bridge(\sigma)$ occurs $n>1$ times in $\head(\sigma)$, then $x$ also occurs $n$ times in the head of $\hgrule(\sigma)$. Accordingly, $x$ occurs with different names (e.g., $x_1,...,x_n$) both in the head and in the body of $\mainrule(\sigma)$. For example, if the ontology contains only the rule $\sigma : P(x) \rightarrow R(x,x)$, then $\hgrule(\sigma): P(x) \rightarrow \mathtt{Aux}_\sigma(x,x) $ and  $\mainrule(\sigma): \mathtt{Aux}_\sigma(x_1,x_2) \rightarrow R(x_1,x_2)$. Clearly, the latter two rules together are equivalent to $\sigma$ w.r.t. to the schema $\{P,R\}$.
We prefer to keep the formal definition of $\hgrule(\Sigma)$ and $\mainrule(\Sigma)$ light without formalising such special cases.

We are now ready to formally introduce the class \dyac.

\begin{definition}[\dyac]\label{def:DyaC}
	Consider a class \C of TGDs such that \certEvalC is decidable. We say that $\Sigma$ belongs to $\dyac$ if $\Sigma$ belongs to $\C$ or if $\mainrule(\Sigma)$ belongs to $\C$.\mybox
\end{definition}

According to the previous definition, one can easily state the following property.

\begin{proposition}\label{prop:containment}
Consider a class \C of TGDs. It holds that $\C \subseteq \dyac$.
\end{proposition}

By Definition \ref{def:DyaC}, to check if an ontology $\Sigma$ belongs to \dyac, one has to verify if $\Sigma \in \C$, or $\mainrule(\Sigma) \in \C$.
We observe that the construction of the set $\mainrule(\Sigma)$ explained above, is polynomial (indeed linear) with respect to the size of $\Sigma$.
Hence, the following result holds.

\begin{theorem}
Consider a class \C of TGDs and assume that checking whether an ontology belongs to $\C$ is doable in some complexity class $\mathbb{C} \supseteq \PTIME$.
Then, checking whether an ontology belongs to $\dyac$ is decidable and it belongs to $\mathbb{C}$\rem{too}.
\end{theorem}

\begin{proof}
We start recalling that, by Definition \ref{def:DyaC}, an ontology $\Sigma$ belongs to $\dyac$ if $(i)$ $\Sigma \in \C$, or $(ii)$ $\mainrule(\Sigma) \in \C$.
Accordingly, checking condition $(i)$ is doable in some complexity class $\mathbb{C} \supseteq \PTIME$, by assumption;
otherwise, the construction of the set $\mainrule(\Sigma)$ is done by a procedure that is polynomial (indeed linear) with respect to the size of $\Sigma$ and, hence, always terminates.
Accordingly, checking condition $(ii)$ is also decidable and doable in the complexity class $\mathbb{C}$.
\end{proof}





The next step is to prove the existence of a \C-dyadic pair for any ontology $\Sigma \in \dyac$.

\begin{theorem}
Consider a set $\Sigma \in \dyac$. There exists  a $\C$-dyadic pair $\Pi = (\sighg, \sigc)$ of TGDs such that $\sighg \cup \sigc \equiv_{\mathbf{\sch(\Sigma)}} \Sigma$. In particular,
\vspace{-2mm}
\begin{itemize}
\item[]
    \begin{enumerate}
        \item If $\Sigma \in \C$, then $\sighg = \emptyset$ and $\sigc = \Sigma$;
        \item If $\Sigma \not\in \C$, then $\sighg = \hgrule(\Sigma)$ and $\sigc = \mainrule(\Sigma)$.
     \end{enumerate}
\end{itemize}
\end{theorem}

\begin{proof}
We claim that for each ontology $\Sigma \in \dyac$ it is possible to construct a \C-dyadic pair $\Pi = (\sighg, \sigc)$ of TGDs such that $\sighg \cup \sigc \equiv_{\mathbf{\sch(\Sigma)}} \Sigma$.

According to Definition~\ref{def:DyadicDec}, we recall that a pair $\Pi = (\sighg, \sigc)$ is \C-dyadic if $(i)$ $\sighg$ is head-ground with respect to $\sighg \cup \sigc$ and $(ii)$ $\sigc \in \C$.
Assume first that $\Sigma \in \C$. Then, trivially the pair $(\emptyset,\Sigma)$ is a \C-dyadic pair.
Assume now that $\Sigma \notin \C$ and let $\Pi = (\hgrule(\Sigma), \mainrule(\Sigma))$.
Property $(ii)$ is satisfied since by hypothesis $\Sigma \in \dyac$; hence, it follows by definition that $\mainrule(\Sigma) \in \C$.
It remains to show property $(i)$.
According to Definition~\ref{def:headGround}, the set $\hgrule(\Sigma)$ has to satisfy four properties.
Property~1 and 2 are trivially fulfilled since, by construction, for each $\sigma \in \Sigma$, $\hgrule(\sigma)$ is a datalog rule and each head atom contains only harmless variables with respect to $\hgrule(\Sigma) \cup \mainrule(\Sigma)$. 
Property~3 and 4 state that \mbox{$\hdpred(\hgrule(\Sigma)) \cap \bdpred(\hgrule(\Sigma)) = \emptyset$} and $ \hdpred(\hgrule(\Sigma)) \cap \hdpred(\mainrule(\Sigma)) = \emptyset $.
These hold since, by construction, $\hdpred(\hgrule(\Sigma)) = \{\mathtt{Aux}_{\sigma} : \sigma \in \Sigma\}$, where each $\mathtt{Aux}_{\sigma}$ is a predicate that does not occur neither in any body of $\hgrule(\Sigma)$ nor in any head of $\mainrule(\Sigma)$.

Concerning the equivalence between $\sighg \cup \sigc$ and $\Sigma$, we can observe that it easily comes from the shape of  $\hgrule(\sigma)$ and $\mainrule(\sigma)$ with respect to each original rule $\sigma \in \Sigma$.
Indeed, first, the body of $\sigma$ is first partitioned in $\satoms(\sigma)$ and $\patoms(\sigma)$.
Second, all the atoms if $\satoms(\sigma)$ form the body of $\hgrule(\sigma)$. 
Then, all the variables of $\hgrule(\sigma)$ that are in join with $\patoms(\sigma)$ or are in the head of $\sigma$ are collected in $\Aux_{\sigma}(\bridge(\sigma))$.
Finally, $\Aux_{\sigma}(\bridge(\sigma))$ is put in conjunction with $\patoms(\sigma)$ to form the body of  $\mainrule(\sigma)$. 
Such a way of decomposing a rule $\sigma$ is well-known to be correct for query answering purposes even when the variables in the auxiliary atom are harmful.
\end{proof}

\begin{algorithm}[t!]
	\DontPrintSemicolon
	\LinesNumbered
	\KwInput{A database $D$, a $\Sigma \in \dyac$, a CQ $q(\mathbf{x})$, and a tuple $\mathbf{c} \in \DC^{|\mathbf{x}|}$}
	$\sighg = \mathit{hg}(\Sigma) $; \\
	$\sigc = \mathit{main}(\Sigma)$; \\
	$\Pi = (\sighg, \sigc)$; \\
	\Return{$\mathsf{DpCertEval}_{[\C]}(q,D,\Pi,\mathbf{c})$};
	
	\caption{$\algCertEvalDyaC(q,D,\Sigma,\mathbf{c})$}\label{alg:CertEvalDyaC}
\end{algorithm}

\noindent It remains to show that \dyac is decidable.
We rely on Algorithm \ref{alg:CertEvalDyaC} together with Theorem \ref{thm:equivCertEval} and Proposition \ref{prop:AlgCorretness} to state the following result.

\begin{theorem}\label{thm:certEvalDyacIsDecidable}
Consider a decidable class $\C$ of TGDs. Then, \certEvalDyaC is decidable.
\end{theorem}

\begin{proof}
To prove the statement we provide the terminating Algorithm \ref{alg:CertEvalDyaC}.
Let $\Sigma \in \dyac$ an ontology. 
Instructions 1 and 2 of the algorithm are introduced in order to construct the components $(\hgrule(\Sigma), \mainrule(\Sigma))$ of a dyadic pair $\Pi$, which is successively initialized at instruction 3.
Of course, the construction of $\Pi$ is based on a polynomial procedure with respect to the size of the input $\Sigma$, hence these instructions always terminates.
Finally, instruction 4 returns the result of the evaluation of the problem $\dpcertEvalC$.
To solve the latter, is invoked Algorithm \ref{alg:DpCertEvalC}, which in turn invokes Algorithm \ref{alg:DBcomplete};
their correctness is guaranteed by Theorem \ref{thm:equivCertEval} and Proposition \ref{prop:AlgCorretness}, respectively.
Accordingly, $\certEvalDyaC$ is decidable.
\end{proof}



\noindent Finally, we conclude the section by proving that, for any class \C including the class \dfinds defined in Section \ref{sec:dfinds}, the class \dyac includes \datalog.

\medskip\medskip

\begin{theorem}\label{thm:DatalogSubsetDyaC}
Consider a class \C of TGDs. If $\C \supseteq \dfinds$, then $\datalog \subseteq \dyac$.
\end{theorem}

\begin{proof}
Consider a datalog rule $\sigma: \Phi(\mathbf{x,y}) \rightarrow P(\mathbf{x})$.
Then, it can be decomposed into the following two rules:
$\hgrule(\sigma): \Phi(\mathbf{x,y}) \rightarrow \mathtt{Aux}_\sigma(\mathbf{x})$ and %
$\mainrule(\sigma): \mathtt{Aux}_\sigma(\mathbf{x}) \rightarrow P(\mathbf{x})$.
Trivially, $\hgrule(\sigma) \in \sighg$;
it remains to show that $\mainrule(\sigma) \in \dfinds$.
According to Definition \ref{def:dfinds}, the full property immediately follows since datalog rules do no have existential variables, whereas the autonomous property holds since $\body(\mainrule(\sigma))$ contains only the fresh predicate $\mathtt{H}$.
In particular, $\mainrule(\sigma) \in \C$, since by assumption $\C \supseteq \dfinds$;
hence, the thesis follows.
\end{proof}

\section{Computational Complexity of Query Answering}\label{sec:complex}
In this section we study the complexity of the $\certEval$ problem over dyadic existential rules.
We start analyzing the data complexity of the problem and then the combined complexity.

\begin{theorem}\label{thm:DataComplexityDyaC}
Consider a class \C of TGDs. 
In data complexity, if $\certEvalC$ belongs to some decidable complexity class $\mathbb{C}$, then the following hold:
\vspace{-2mm}
\begin{itemize}
\item[]
    \begin{enumerate}
        \item If  $\mathbb{C} \subseteq \PTIME$, then $\certEvalDyaC$ is in \PTIME;
        
        \item If  $\mathbb{C} \supseteq \PTIME$, then $\certEvalDyaC$ is in $\PTIME^{\mathbb{C}}$;
                
        \item If  $\mathbb{C} \supseteq \PTIME$ is deterministic%
        \footnote{By \vir deterministic" we mean that $\certEvalDyaC$ can be solved by a deterministic Turing machine. We refer to the classes $\PTIME, \PSPACE, \iEXPTIME, \iEXPSPACE$, for all $i \geq 1$.}
        and $\certEvalC$ is $\mathbb{C}${\footnotesize\textrm{-complete}}, then it holds that $\certEvalDyaC$ is $\mathbb{C}${\footnotesize\textrm{-complete}} too;
        
        \item If $\mathcal{C} \supseteq \dfinds$, then
        $\certEvalDyaC$ is $\PTIMEh$.
    \end{enumerate}
\end{itemize}
\end{theorem}

\begin{proof}
To prove the theorem, we rely on the complexity of Algorithm \ref{alg:CertEvalDyaC}.
Let $D$ be a database, $\Sigma \in \dyac$ an ontology, $q(\mathbf{x})$ a CQ, and $\mathbf{c} \in \DC^{|\mathbf{x}|}$ a tuple.
Moreover, let $d = |\const(D)|$ and $\mu = \max_{P \in \hdpred(\sighg)} \arity(P)$.
As shown in Theorem \ref{thm:certEvalDyacIsDecidable}, Algorithm \ref{alg:CertEvalDyaC} terminates and correctly decides $\certEvalDyaC$.
More specifically, instructions 1 and 2 are introduced in order to construct, respectively, the first component $\hgrule(\Sigma)$ and the second component $\mainrule(\Sigma)$ of a dyadic pair.
The procedure used to build these sets
is polynomial (indeed linear) with respect to the size of the ontology $\Sigma$ in input.
By neglecting the \vir trivial" instruction 3, the computational cost of the Algorithm mainly depends on instruction 4, that is on the invocation of Algorithm \ref{alg:DpCertEvalC}, which in turn invokes Algorithm \ref{alg:DBcomplete} in order to compute the completed database $\dcompl$.
As stated in Proposition \ref{prop:AlgCorretness}, Algorithm \ref{alg:DBcomplete} always terminates.
In particular, the size of the set $\dcompl$ is at most $|\hdpred(\sighg)| \cdot d^{\mu}$; the problem $\certEvalC$, at instruction 5, is called $d^{\mu}$ times, and the \textbf{for}-loop at instruction 3 is executed $|\sighg|$ times.
Therefore, by ignoring the computational costs of the oracle (i.e., checking whether $\mathbf{t} \in \cert(q,\dcompl,\sigc)$), Algorithm \ref{alg:DBcomplete} overall performs a number of step that is linear in $|\sighg| \cdot |\hdpred(\sighg)| \cdot d^{2\mu}$. Indeed, this value is also an upper bound for the number of calls to the oracle.
Since we are in data complexity, the following parameters are bounded: the maximum arity $\mu$, the size of the sets $\sighg$ and $\sigc$, as well as the size and the number of each query $q$ constructed at instruction 4 of Algorithm \ref{alg:DBcomplete}. Hence, the latter calls polynomially many times the problem $\certEvalC$.
Accordingly, Algorithm \ref{alg:CertEvalDyaC} is polynomial and in turn it invokes polynomially many times an oracle to compute \certEvalC.
Hence, if $\mathbb{C} \subseteq \PTIME$, trivially $\certEvalDyaC \in \PTIME$; whereas, if $\mathbb{C} \supseteq \PTIME$,   $\certEvalDyaC \in \PTIME^{\mathbb{C}}$.
To prove point 3 of the theorem, we observe that the membership follows from point 2 and from the fact that, for any deterministic class $\mathbb{C} \supseteq \PTIME$, it holds that $\PTIME^{\mathbb{C}} = \mathbb{C}$;
whereas, the hardness derives from Proposition \ref{prop:containment}, since $\dyac$ includes the class \C, that is $\mathbb{C}${\footnotesize\textrm{-hard}} by assumption.
Finally, to prove point 4, we recall that by Theorem \ref{thm:DatalogSubsetDyaC}, $\datalog \subseteq \dyac$; hence, since $\certEval[\datalog]$ is $\PTIMEh$, the thesis follows.  
\end{proof}

\noindent Accordingly to the above theorem, immediately we get the following result.

\begin{table}[t!]
	\centering
	\caption{Data complexity comparison of $\certEvalC$ with $\certEvalDyaC$. }\label{table:DataComplexityComparison}
	{\tablefont\begin{tabular}{@{\extracolsep{\fill}}lcc}
		\topline
		Class \C & \C & \dyac
		\midline
		$\wfrguarded$
		& $\EXPTIMEc$ & $\EXPTIMEc$ \\
		$\frguarded$
		& $\PTIMEc$ & $\PTIMEc$ \\
		$\weaklyacy$
		& $\PTIMEc$ & $\PTIMEc$ \\
		$\joinacy$
		& $\PTIMEc$ & $\PTIMEc$ \\
		$\datalog$
		& $\PTIMEc$ & $\PTIMEc$ \\
		$\shy$
		& $\PTIMEc$ & $\PTIMEc$ \\
		$\ward$
		& $\PTIMEc$ & $\PTIMEc$ \\
		$\prot$
		& $\PTIMEc$ & $\PTIMEc$ \\
		%
		$\stickyj$
		& $\ACzero$ & $\PTIMEc$ \\
		$\linear$
		& $\ACzero$ & $\PTIMEc$ \\
		$\joinless$
		& $\ACzero$ & $\PTIMEc$ \\
		$\indep$
		& $\ACzero$ & $\PTIMEc$ 
	\botline
\end{tabular}}
\end{table}

\begin{corollary}
Complexity results in Table~\ref{table:DataComplexityComparison} do hold.
\end{corollary}

For studying the combined complexity, we need to take into account the fact that the database returned by Algorithm~\ref{alg:DBcomplete} (namely, $\dcompl$) is exponential with respect to the input one (namely, $D$). 
Indeed, the check $\mathbf{c} \in \cert(q,\dcompl,\sigc)$ performed by  Algorithm~\ref{alg:DpCertEvalC} is done on an exponentially bigger database.
Thus, in case $\certEvalC$ would have the same data complexity and combined complexity, it might happen that the combined complexity of $\certEvalDyaC$ could be exponentially higher that the one of $\certEvalC$.
Although all the considered classes in $\mathbb{E}^+_{\mathit{syn}}$ do not suffer from this shortcoming, before stating our general result, we need to focus on \vir well-behaved'' classes of TGDs.
A class $\C$ of TGDs enjoys the {\em dropping data-complexity} property if there is an exponential jump from the combined complexity of $\certEvalC$ to the data complexity of $\certEvalC$.

\begin{proposition}
Each class in $\mathbb{E}_{\mathit{syn}}$ enjoys the dropping data-complexity property.
\end{proposition}

We can now state the last result of the section, providing the combined complexity of problem $\certEval$ over $\dyac$ sets of TGDs.

\begin{theorem}
Consider a class \C of TGDs. 
In combined complexity, if $\certEvalC$ belongs to some decidable complexity class $\mathbb{C}$ and $\C$ enjoys the dropping data-complexity property, then the following hold:
\vspace{-2mm}
\begin{itemize}
\item[]
    \begin{enumerate}
        \item If  $\mathbb{C} \subseteq \EXPTIME$, then $\certEvalDyaC$ is in \EXPTIME;
        
        \item If  $\mathbb{C} \supseteq \EXPTIME$, then $\certEvalDyaC$ is in $\EXPTIME^{\mathbb{C}}$;
        
        \item If  $\mathbb{C} \supseteq \EXPTIME$ is deterministic and $\certEvalC$ is $\mathbb{C}${\footnotesize\textrm{-complete}},
        then it holds that $\certEvalDyaC$ is $\mathbb{C}${\footnotesize\textrm{-complete}} too;
        
        \item If $\mathcal{C} \supseteq \dfinds$, then
        \certEvalDyaC is \EXPTIME{\footnotesize\textrm{-hard}}.
    \end{enumerate}
\end{itemize}
\end{theorem}

\begin{table}[b!]
	\centering
	\caption{Combined complexity comparison of $\certEvalC$ with $\certEvalDyaC$. }\label{table:CombinedComplexityComparison}
	{\tablefont\begin{tabular}{@{\extracolsep{\fill}}lcc}
		\topline
		Class \C & \C & \dyac
		\midline
		$\wfrguarded$
		& $\TwoEXPTIMEc$ & $\TwoEXPTIMEc$  \\
		$\frguarded$
		& $\TwoEXPTIMEc$ & $\TwoEXPTIMEc$ \\
		$\weaklyacy$
		& $\TwoEXPTIMEc$ & $\TwoEXPTIMEc$ \\
		$\joinacy$
		& $\TwoEXPTIMEc$ & $\TwoEXPTIMEc$ \\
		$\datalog$
		& $\EXPTIMEc$ & $\EXPTIMEc$ \\
		$\shy$
		& $\EXPTIMEc$ & $\EXPTIMEc$ \\
		$\ward$
		& $\EXPTIMEc$ & $\EXPTIMEc$ \\
		$\prot$
		& $\EXPTIMEc$ & $\EXPTIMEc$ \\
		%
		$\stickyj$
		& $\EXPTIMEc$ & $\EXPTIMEc$ \\
		$\linear$
		& $\PSPACEc$ & $\EXPTIMEc$ \\
		$\joinless$
		& $\PSPACEc$ & $\EXPTIMEc$ \\
		$\indep$
		& $\PSPACEc$ &  $\EXPTIMEc$
	\botline
\end{tabular}}
\end{table}

\begin{proof}
The argument proceeds similarly to proof of Theorem \ref{thm:DataComplexityDyaC} by arguing on Algorithm \ref{alg:CertEvalDyaC} to determine the complexity of $\certEvalDyaC$.
Let $D$ be a database, $\Sigma \in \dyac$ an ontology, $q(\mathbf{x})$ a CQ, and $\mathbf{c} \in \DC^{|\mathbf{x}|}$ a tuple.
Moreover, let $d = |\const(D)|$ and $\mu = \max_{P \in \hdpred(\sighg)} \arity(P)$.
As previously shown, Algorithm \ref{alg:CertEvalDyaC} invokes Algorithm \ref{alg:DpCertEvalC}, which in turn invokes Algorithm \ref{alg:DBcomplete}.
Concerning the latter, by ignoring the computational costs of the oracle, it overall performs a number of step that is linear in $|\sighg| \cdot |\hdpred(\sighg)| \cdot d^{2\mu}$. Indeed, also in this case, this value is an upper bound for the number of calls to the oracle.
\change{This is enough to show point $2$.} 

\change{Concerning the memberships of point $1$ and point $3$,}
differently from the proof of Theorem \ref{thm:DataComplexityDyaC}, in combined complexity the maximum arity $\mu$, the size of the sets $\sighg$ and $\sigc$, as well as the size and the number of each query constructed at instruction 4 of Algorithm \ref{alg:DBcomplete} are not bounded.
Accordingly, also the size of the completed database returned by Algorithm \ref{alg:DBcomplete} (namely $\dcompl$) \rem{becomes}
\change{may become} exponential with \change{respect to} the input.
\change{More precisely, $|\const(\dcompl)| = d$ and $|\dcompl| \leq |\hdpred(\sighg)|\cdot d^\mu + |D|$. 
Let $n$ generically denote the size $||seq||$ of any sequence  $\mathit{seq}$ of objects given in input to $\certEvalC$.
We can now consider the cost function $g(n)$ (resp., $f(n)$) of some algorithm/oracle that decides $\certEvalC$ and shows that it belongs to $\mathbb{C}$ (resp., $\mathbb{C}_d$) 
in combined (resp., data) complexity.
According to the dropping data-complexity property, we know that $g(n)$ grows at least exponentially faster than $f(n)$.
Essentially, there is an exponential jump from $\mathbb{C}_d$ to $\mathbb{C}$ that does not depend on the size of the input database but only on the size of other parameters, namely the ontology, the query and the tuple of constants.
Consider now the query
$q':= \langle \mathbf{x} \rangle \leftarrow \Phi(\mathbf{x,y})$ constructed at instruction 4 of Algorithm~\ref{alg:DBcomplete} (we call it $q'$ to avoid confusion with $q(\mathbf{x})$ mentioned at the beginning of this proof).
At instruction 5 of the same algorithm, the oracle for $\certEvalC$ checks whether $\mathbf{t} \in \cert(q',\dcompl,\sigc)$ holds.
Since $g(n)$ grows at least exponentially faster than $f(n)$, we get that  $g(||\mathbf{t},q',\dcompl,\sigc||)$ remains of the same exponential order of $g(||\mathbf{t},q',D,\sigc||)$, although $||D^+||$ may be  exponentially larger than $||D||$.}
\rem{the dependence on the database, in the cost-function of $\certEvalC$, is at least exponentially smaller than other parameters appearing in the cost-function and, 
Thus, also bloating exponentially the database size the combined complexity
of $\certEvalDyaC$ is not exponentially higher that the one of $\certEvalC$.}
%
%
%
%
%

Regarding point $1$, if $\mathbb{C} \subseteq \EXPTIME$, then 
\change{we know that a $\mathbb{C}$-oracle for $\certEvalC$ works at most exponentially in $||\mathbf{t},q',\sigc||$ and  at most polynomially (resp., exponentially) in $||\dcompl||$ (resp., $||D||$); thus, in this  case, the $\mathbb{C}$-oracle cannot reach  double-exponential time but it remains exponential. Therefore, $\certEvalDyaC \in \EXPTIME$.}

%
\change{For the memberships of point 3, we already know that $\certEvalDyaC$ is in $\EXPTIME^{\mathbb{C}}$. Consider now a
$\mathbb{C}$-oracle $O$ for $\certEvalC$ characterised by the cost function $g(n)$.
If $\mathbb{C} \supseteq \EXPTIME$ is  deterministic, then $O$ works with respect to $||\mathbf{t},q',\sigc||$ in an exponentially faster way than with respect to $||\dcompl||$; thus, also in this  case, $O$ cannot exceed the power of $\mathbb{C}$. Therefore, $\EXPTIME^{\mathbb{C}}$, in a sense, collapses to $ \mathbb{C}$.}

Finally, we conclude the proof by considering the hardness of points 3 and 4. 
In the first case, we observe that it derives from Proposition \ref{prop:containment}, since $\dyac$ includes the class \C, that is $\mathbb{C}${\footnotesize\textrm{-hard}} by assumption.
For point 4, we recall that by Theorem \ref{thm:DatalogSubsetDyaC}, $\datalog \subseteq \dyac$; hence, since $\certEval[\datalog]$ is $\EXPTIME${\footnotesize\textrm{-hard}}, it follows the thesis. \end{proof}

The following immediately derives from above theorem.

\begin{corollary}
Complexity results in Table~\ref{table:CombinedComplexityComparison} do hold.
\end{corollary}

\section{Conclusion}
Dyadic decomposable sets form a novel decidable class of TGDs that encompasses and generalises all the existing (syntactic and semantic) decidable classes of TGDs. 
In the near feature, it would be interesting to implement a prototype for dyadic existential rules by exploiting different kinds of existing reasoners.

%
%

\bibliographystyle{tlplike}
\bibliography{bibtex_new}

\end{document}